\documentclass[a4paper,11pt]{article}
\usepackage{jheppub} % for details on the use of the package, please see the JINST-author-manual
\usepackage[utf8]{inputenc}
\usepackage{graphicx}
\usepackage{subfigure}
\usepackage{setspace}
\usepackage{amsmath}
\usepackage{commath}
\usepackage{amsthm}
\usepackage{color}
\usepackage{braket}
\usepackage{color,soul}
\usepackage{appendix}
\usepackage{amssymb}
\allowdisplaybreaks
\usepackage[dvipsnames]{xcolor}

\usepackage{lineno}
\newtheorem{proposition}{Proposition}

\newtheorem{theorem}{Theorem}

\newtheorem{corollary}{Corollary}[theorem]
\DeclareMathOperator{\Tr}{Tr}
\newcommand{\oper}[1]{\ket{#1}\bra{#1}}

\title{Distribution of quantum gravity induced entanglement in many-body systems}

\author[a, b]{Pratik Ghosal,}
\author[c, d]{Arkaprabha Ghosal,}
\author[a]{Somshubhro Bandyopadhyay} %\footnote{Corresponding author: Somshubhro Bandyopadhyay (som@jcbose.ac.in)}}
\affiliation[a]{Department of Physical Sciences, Bose Institute, EN 80, Sector V, Bidhannagar, Kolkata 700091, India}
\affiliation[b]{Department of Physics of Complex Systems, S. N. Bose National Centre for Basic Sciences, Block JD, Sector III, Salt Lake, Kolkata 700106, India}
\affiliation[c]{Optics \& Quantum Information Group, The Institute of Mathematical Sciences, CIT Campus, Taramani, Chennai 600 113, India}
\affiliation[d]{Homi Bhabha National Institute, Training School Complex, Anushakti Nagar, Mumbai 400085, India}

\emailAdd{ghoshal.pratik00@gmail.com}
\emailAdd{a.ghosal1993@gmail.com}
\emailAdd{som@jcbose.ac.in}

\abstract{Recently, it was shown that if two distant test masses, each in a spatially superposed quantum state, become entangled due to their mutual gravitational interaction, then this entanglement could serve as evidence of the quantum nature of gravity. We extend this treatment to a many-body system in a general setup and study the entanglement properties of the time-evolved state. We exactly compute the time-dependent I-concurrence for every bipartition and obtain the necessary and sufficient condition for the creation of genuine many-body entanglement. We further show that this entanglement is of generalised GHZ type when certain conditions are met. We also evaluate the amount of multipartite entanglement in the system using a set of generalised Meyer-Wallach measures.}

\begin{document}
\maketitle
\flushbottom
%--------------------------------------------------------------

\section{Introduction}
%-----------------------------------------------------------------

Gravity remains the only known interaction in nature for which we do not have any empirical evidence of it being quantum. In particular, verifying any reasonable prediction of diverse quantum theories of gravity would require energy as high as Planck energy---the regime where it is commonly believed that the quantum effects of gravity are revealed. Such high energies, however, are not accessible with current technologies. So, alternatively, one could look for gravity-induced phenomena at lower energies that could only be explained by treating gravity as a quantum entity.

Recently, a table-top experiment, namely quantum gravity-induced entanglement of masses (QGEM), has been proposed to test the quantum nature of gravity \cite{Bose2017PRL, Marletto-Vedral2017PRL}. The proposal considers two distant neutral test masses, each prepared in a superposition of spatially localised quantum states, interacting via their mutual gravity. It then shows that the time evolution under gravity leads to entanglement between the masses, which can be witnessed using test masses with embedded spins in a Stern-Gerlach (SG) interferometric setup through spin-correlation measurements. This entanglement, it was argued, is a signature of the quantum nature of gravity. Though challenges remain, the experimental realisation of this proposal may not be far away in the foreseeable future.

The argument \cite{Bose2017PRL} relies on two key elements: the gravitational interaction between the test masses is mediated by a field, and the quantum information principle: \textit{no creation of entanglement by local operations and classical communication (LOCC)} \cite{PhysRevA.54.3824, Horodecki2009EntanglementReview}, which states that two or more spatially separated unentangled quantum systems interacting only via classical channels cannot become entangled. The argument now goes as follows: Since the test masses were not entangled initially, they could never become entangled if the gravitational field is classical. The creation of entanglement thus implies that the mediator of the gravitational interaction, assumed to be massless spin-2 graviton, must be quantum, and the exchange of this virtual graviton is a form of quantum communication responsible for creating entanglement \cite{Marshman2020}.\footnote{Ref. \cite{Christodoulou:2018cmk} argues that this entanglement could also be seen as evidence of the quantum superposition of different spacetime geometries (metrics).} Note that, for the argument to go through, one only needs to ensure that gravity is the only dominant interaction between the masses, which requires the distance between the test masses to be such that even for their closest approach, the short-range Casimir-Polder interaction is negligible \cite{Bose2017PRL}.

The entanglement generation phenomenon has been studied within the framework of linearised quantum gravity  \cite{Marshman2020, PhysRevD.105.106028, PhysRevLett.130.100202}. However, it has been demonstrated that the primary contribution originates solely from the Newtonian potential term \cite{Marshman2020, PhysRevLett.130.100202}. This interestingly shows that even the Newtonian gravitational interaction can exhibit quantum characteristics. Moreover, it has been shown that the entanglement generated by Newtonian interaction necessarily implies the existence of gravitons \cite{PhysRevD.105.024029}, with little apparent distinction observed between entanglement mediated by Newtonian gravity and that mediated by gravitons  \cite{PhysRevD.105.086001}.

It is important to note that several features of gravity-induced entanglement, especially entanglement generation rate, depend on the setup, which in a two-body scenario is determined by the directions of the superpositions. By ``direction of superposition," we refer to the direction from one spatially localised state to another within the superposition of a single mass. If the superpositions are in the same direction as their separation, the setup is linear (as in the original proposal \cite{Bose2017PRL}), whereas in a parallel setup, the superpositions are parallel to each other and perpendicular to the separation. Some setups are better than others; for example, the entanglement generation rate in a  parallel setup is higher than linear \cite{nguyen2020entanglement}.

Subsequent works have considered quantum aspects of gravitational interaction \cite{Marshman2020, Christodoulou:2018cmk, PhysRevD.105.106028, PhysRevLett.130.100202}, entanglement generation in optomechanical setups \cite{PhysRevD.102.106021, Krisnanda2020, Carney_2019, Carlesso_2019} and atom interferometer \cite{PRXQuantum.2.030330} (see also \cite{universe8020058}), an entanglement witness that reduces the required interaction time, thereby making experiments feasible for high decoherence rates \cite{PhysRevA.102.022428}, noise and decoherence due to gravitons \cite{PhysRevD.103.044017}, and other decoherence effects in experiments \cite{Rijavec_2021}.

\subsection*{Motivation and Summary} A notable feature of the QGEM framework is that it can be extended, without much difficulty, to more than two test masses. Since gravity acts pairwise, one would expect the time evolution, in general, would lead to a many-body entangled state (see also \cite{Haine_2021}). Indeed, recent results have shown this to be the case: Ref. \cite{Miki2021} considered $N\geq 3$ test masses in a finite linear chain configuration and studied decoherence effect due to gravitational interaction, Ref. \cite{PhysRevD.107.064054} considered both three and four test masses and identified the setup that is most efficient for entanglement generation rate, Ref. \cite{Schut2022} compared entanglement generation rate and entanglement robustness of multiple setups involving three test masses, and Ref. \cite{liu2023multiqubit} investigated multipartite entanglement in symmetric geometries. These works, however, have considered specific setups, namely, linear, parallel, or a star \cite{Schut2022}.

The motivation of the present work stems from the point of view that (quantum) gravity is a bona fide means to entangle massive objects. Then, what are the characteristic features of this entanglement and quantum correlations in a many-body system?  It is expected that the setup of the system, i.e., the relative orientations of the mass superpositions, will play a crucial role in determining the characteristics of the many-body entanglement. Therefore, to answer our question one needs to consider the time evolution in a general setup, that is, without any underlying assumptions on the geometry or superposition directions. In this article, we study the properties of gravity-induced many-body entanglement with as much generality as possible. Since many-body entanglement in quantum systems is rich in structure and is known to exhibit features not present in the two-body scenario, we expect to identify the properties that stand to hold when the entanglement arises from mutual gravitational interactions between massive objects.

Specifically, we consider a system of $N\geq 3$ neutral test masses, each initially prepared in a superposition of two spatially localised states. We do not assume any particular geometry associated with the system, nor do we assume any specific setup; that is, the relative orientations of the superpositions are, in general, arbitrary. However, one thing needs to be ensured: for any conceivable setup, the distance between any two localised states, each belonging to a different superposition, should be such that the Casimir-Polder force is negligible and gravity is the only dominant interaction. This implies that while the localised states from different superpositions need to be close, they cannot be too close to each other.

Since each test mass is initially prepared in a superposition of two spatially localised states, it is associated with two spatial degrees of freedom and can thus be considered effectively as a qubit---a two-level quantum system. Therefore, we start from an $N$-qubit product state that evolves under a unitary time evolution operator, where the complete Hamiltonian is the sum of pairwise Hamiltonians, each corresponding to the gravitational interaction between a mass-pair. The purpose of the present paper is to study the properties of the time-evolved $N$-qubit pure state.

To study the entanglement properties, we will primarily consider the bipartition approach. In this approach, one looks at entanglement across every bipartition and tries to understand how it is distributed throughout the system. For example, we say the system is genuine many-body entangled if and only if entanglement is nonzero across every bipartition. Generally speaking, unless there are enough symmetries, one needs to look at all bipartitions to get a reasonably complete picture of entanglement distribution. To compute the multipartite entanglement present in the system, we use a set of generalised Meyer-Wallach measures \cite{PhysRevA.69.052330, meyer2002global}.

We begin with the two-body scenario in a general setup (Section \ref{s2}). We obtain an expression of the time-evolved two-body state by extending the prescription outlined in Ref. \cite{Bose2017PRL}. We then compute its concurrence \cite{PhysRevLett.80.2245}, which is found to be a function of time and the entangling phase. This entangling phase, which fully determines the entanglement, is only a function of distances between the localised states from different superpositions. So, one could have different setups but the same entangling phase and, therefore, identical entanglement properties. The entanglement exhibits periodic oscillations, the period being inversely proportional to the entangling phase. Thus, our analysis shows that this behaviour is generic.\footnote{All setups, except those for which the entangling phase is zero, generate entanglement which shows this oscillatory behaviour.} We also introduce the notations and discuss the unitary description of the dynamics in Section \ref{s2}.

In Section \ref{s3}, we discuss the three-body problem. Here, we consider the geometry in which the test masses are placed on the vertices of a triangle with unequal sides, and the directions of the superpositions are arbitrary. This is clearly the most general setup conceivable. Since the time-evolved state is a three-qubit pure state, it can be written as a pure state in $\mathbb{C}^{2}\otimes \mathbb{C}^{4}$ across any bipartition. By the Schmidt decomposition, this state can be treated as a two-qubit pure state. We obtain exact expressions of concurrence \cite{PhysRevLett.80.2245} of this state for every bipartition. In particular, for a bipartition, say $m_i|m_jm_k$, the corresponding concurrence is a function of time and the entangling phases corresponding to the mass-pairs ($m_i,m_j$) and ($m_i, m_k$) but not ($m_j,m_k$). We show that time evolution leads to genuine three-body entanglement if and only if at least two of the three entangling phases are nonzero. In addition, if their ratio is that of two odd integers, the system evolves into a maximally entangled GHZ state with 1 ebit of entanglement across every bipartition. Furthermore, if all three entangling phases are nonzero such that the ratio of any two is not rational, the system becomes genuine three-body entangled and remains so at all times $t > 0$, never becoming bi-separable or fully separable. We also investigate three-body correlations using entanglement monogamy relations \cite{CKW2000}, which tell us about the distribution of pairwise entanglement and three-body correlations as quantified by 3-tangle \cite{CKW2000}. We evaluate the 3-tangle and pairwise entanglements in terms of the three entangling phases.

In Section \ref{s4}, we extend the three-body analysis to an $N$-body system. In considering the most general setup, we do not assume any specific geometry for the arrangement of the masses, nor do we assume any particular orientation of the superpositions of the masses. We are interested in investigating entanglement distribution in the $N$-body system. To this end, we compute an exact expression of the I-concurrence, a measure of bipartite entanglement \cite{Rungta2001}, of the time-evolved $N$-body state for any $k|(N-k)$ bipartition of the system, where $1\leq k\leq \left\lfloor \frac{N}{2}\right\rfloor$. The I-concurrence is given by a function of the entangling phases associated with only those mass-pairs whose masses belong to the alternative parts of the bipartition under consideration. For clarity, we represent the $N$-body system as a graph, with its nodes representing the masses and the edges representing entangling phases, and provide a pictorial representation of the mathematical expression for the I-concurrence. We show that the time-evolution leads to genuine $N$-body entanglement if and only if the graph is ``connected". Additionally, if all the entangling phases are nonzero and the ratio of any two is a ratio of two odd integers, the time evolution periodically leads to a generalised GHZ state with 1 ebit entanglement across every bipartition. We also evaluate a family of generalised Meyer-Wallach measures, quantifying the multipartite entanglement \cite{PhysRevA.69.052330} as a function of all the entangling phases.

In Section \ref{s5}, we compare the entanglement in a given $1|(N-1)$ bipartition across different values of $N$ by varying the number of masses in the part containing $(N-1)$ masses, without changing any other parameters. We show that the I-concurrence for the conisdered bipartition, at any instant of time $t>0$, is lower bounded by the maximum I-concurrence value among the $1|(N-2)$ bipartitions of the corresponding $(N-1)$-body system, each involving one mass less.  This observation suggests a general trend wherein the entanglement shared by a mass with the rest of the system tends to persist for longer durations in $N$-body systems compared to $(N-1)$-body systems.

%----------------------------------------------------------------------
\section{\label{s2} Two-body QGEM: General setup and entangling phase}
%----------------------------------------------------------------------

Consider two neutral test masses $m_1$ and $m_2$, each initially prepared in a spatially superposed state [Fig. \ref{qgem_two}]
\begin{align}
    |\psi\rangle_p=\frac{1}{\sqrt{2}} \sum_{j_p=0_p,1_p} |j_p\rangle,~~~p=1,2,
\end{align}
where $|0_p\rangle$, $|1_p\rangle$ are localised (non-overlapping) Gaussian wavepackets satisfying $\langle 0_p|1_p\rangle=0$. This is ensured by assuming the width of each wavepacket $\ll \delta$, where $\delta$ is the separation between the centres of the wavepackets. The centers of superpositions of the test masses, i.e., the centres of the states $|\psi\rangle_1$ and $|\psi\rangle_2$, are separated by a distance $d$. 

\begin{figure}[t]
\centering
\includegraphics[width=0.6\textwidth]{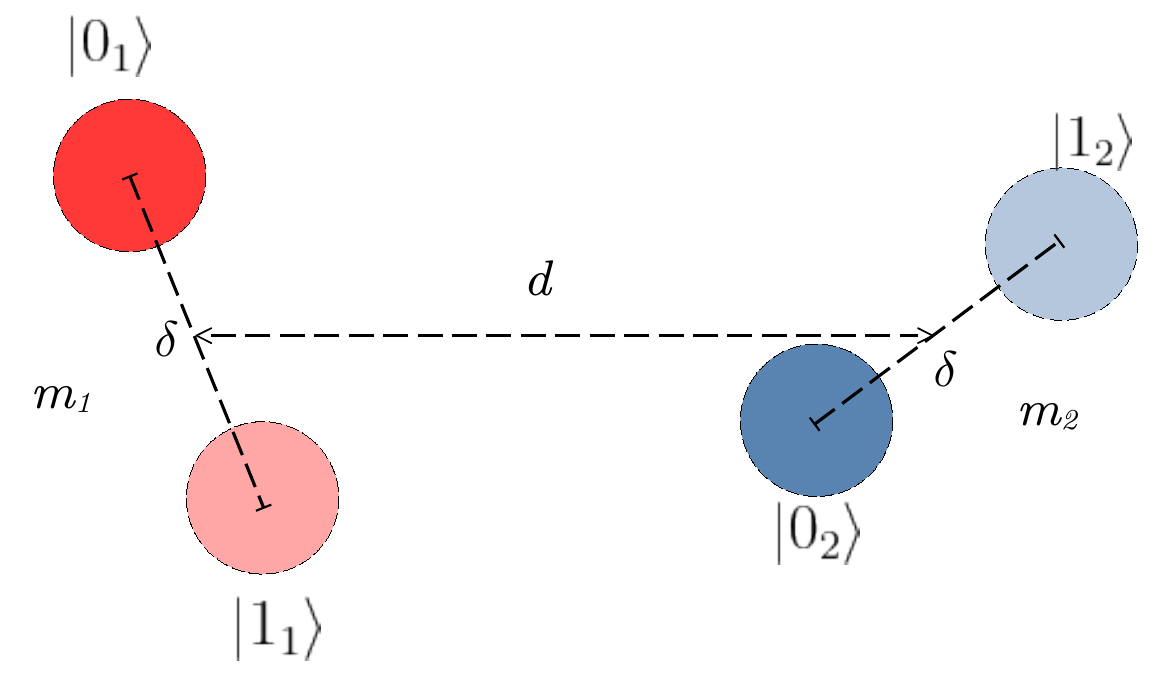}
\caption{\label{qgem_two} Two masses $m_1$ and $m_2$, each prepared in a superposition of two spatially localised states, are kept at a distance $d$ apart from each other. The two localised states of mass $m_p$ are denoted as $|0_p\rangle$ and $|1_p\rangle$, where $p=1,2$, and they have a separation of $\delta$ between them. The orientation of the superpostions, or the setup of the system, is taken to arbitrary.}
\end{figure}

The relative orientation of the lines joining the localised states of each test mass defines the setup of the two-body system. For instance, if the localised states $|0_1\rangle,|1_1\rangle,|0_2\rangle$ and $|1_2\rangle$ all lie on the same straight line, the setup is called linear \cite{Bose2017PRL}. Whereas, if the line joining $|0_1\rangle$ and $|1_1\rangle$ is parallel to the line joining $|0_2\rangle$ and $|1_2\rangle$, and both lines are perpendicular to the line connecting the centres of superpositions, the setup is called parallel \cite{nguyen2020entanglement}. More general setups are also possible, where the line segments joining the localised states are neither parallel nor lie along the line joining the centres of superpositions. We do not fix any specific setup in our analysis.

Let $d_{j_1j_2}$ denote the separation between the states $|j_1\rangle$ and $|j_2\rangle$ for $j_1=0_1,1_1$ and $j_2=0_2,1_2$. We assume that gravity is the only dominant interaction between the masses, even at closest approach, given by $\min_{j_1,j_2} d_{j_1j_2}$, so that the short-range Casimir-Polder force is negligible \cite{Bose2017PRL}. 

The initial state of the system is given by
\begin{align}
    |\Psi(t=0)\rangle_{12} = |\psi\rangle_1|\psi\rangle_2 = \frac{1}{2}\sum_{\substack{j_1=0_1,1_1\\j_2=0_2,1_2}} |j_1\rangle|j_2\rangle. 
\end{align}
Now when the system evolves under gravity, the components of the term $|j_1\rangle|j_2\rangle$ interact under the gravitational potential $V_{j_1j_2} = -\frac{Gm_1m_2}{d_{j_1j_2}}$, which, in general, is different for different $j_1$ and $j_2$. Consequently, $|j_1\rangle|j_2\rangle$ picks up a phase:
\begin{align}
    |j_1\rangle|j_2\rangle \to e^{i\phi_{j_1j_2}t}|j_1\rangle|j_2\rangle,
\end{align}
where $\phi_{j_1j_2}=-V_{j_1j_2}/\hbar$. The state of the system at a later time $t$ is then given by
\begin{align}\label{time-evolved-two}
    |\Psi(t)\rangle_{12} = \frac{1}{2}\sum_{\substack{j_1=0_1,1_1\\j_2=0_2,1_2}} e^{i\phi_{j_1j_2}t}|j_1\rangle|j_2\rangle.
\end{align}

Note that Newtonian potential is used for the gravitational interaction. One could argue that Newtonian gravity is nonlocal and entanglement generated by such an interaction cannot certify the quantumness of the interaction. However, we could have also examined the time evolution using a linearized quantum gravity approach \cite{Marshman2020, PhysRevD.105.106028, PhysRevLett.130.100202}, where the gravitational interaction is local, and it has been shown that the leading contribution to entanglement generation still arises from the Newtonian potential term. Therefore, we will use, following \cite{Bose2017PRL}, the Newtonian approximation of the gravitational interaction.

The entanglement monotone, concurrence, of any two-qubit pure state $|\beta\rangle_{AB}$ is defined as \cite{PhysRevLett.80.2245}
\begin{align}
    \mathbf{C}[|\beta\rangle_{AB}]:= \sqrt{2\left(1-\Tr \rho_A^2\right)},
\end{align}
where $\rho_A=\Tr_{B}|\beta\rangle_{AB}\langle\beta|$ is the marginal state of $A$, obtained by taking partial trace over $B$. Particularly, for $|\Psi(t)\rangle_{12}$, the concurrence can be explicitly computed as
\begin{align}\label{2body_concurrence}
    \mathbf{C}_{12}(t)=\sqrt{1-\cos^2\left(\frac{\Phi_{12}t}{2}\right)},
\end{align}
where
\begin{align}\label{entangling_phase}
    \Phi_{12}=\vert\phi_{0_11_2}+\phi_{1_10_2}-\phi_{0_10_2}-\phi_{1_11_2}\vert
\end{align}
is called the \textit{entangling phase}. From (\ref{2body_concurrence}) it follows that $|\Psi(t)\rangle_{12}$ is entangled if
\begin{align}
    \Phi_{12}t \neq 2n\pi,~n=0,1,2,\ldots
\end{align}
So we see that the entangling phase completely characterises the concurrence function, and hence, entanglement. Further note that entanglement (concurrence) oscillates periodically with time period $T=\frac{2\pi}{\Phi_{12}}$. In particular, the states becomes product at times $t=\frac{2n\pi}{\Phi_{12}}$ and maximally entangled at times $t=\frac{(2n+1)\pi}{\Phi_{12}}$, for $n=0,1,2,\ldots$. The oscillatory behaviour is generic.

The entangling phase is determined by the setup. Using $\phi_{j_1j_2}=-V_{j_1j_2}/\hbar$ and $V_{j_1j_2} = -\frac{Gm_1m_2}{d_{j_1j_2}}$ in (\ref{entangling_phase}) we get 
\begin{align}
    \Phi_{12}=\frac{Gm_1m_2}{\hbar}\left\vert\frac{1}{d_{0_11_2}}+\frac{1}{d_{1_10_2}}-\frac{1}{d_{0_10_2}}-\frac{1}{d_{1_11_2}}\right\vert.
\end{align}
Note however that different setups may yield the same value for $\Phi_{12}$, and in such cases the corresponding
time-evolved states will exhibit identical entanglement properties.

%--------------------------------------------
\subsection*{Unitary description}
%--------------------------------------------

The dynamics of entanglement generation can be effectively described by unitary time evolution of the system. The unitary time-evolution operator is given by
\begin{align}
    \mathcal{U}_{12}(t)=\exp \left(-\frac{i\mathcal{H}_{12}t}{\hbar}\right),
\end{align}
where
\begin{align}\label{Hamiltonian}
    \mathcal{H}_{12}= - Gm_1m_2\frac{1}{\vert\mathcal{\textbf{x}}_1-\mathcal{\textbf{x}}_2\vert}
\end{align}
is the time-independent Hamiltonian $\mathcal{H}_{12}$ and $\mathcal{\textbf{x}}_1$ and $\mathcal{\textbf{x}}_2$ are the position operators associated with $m_1$ and $m_2$ respectively (also see e.g. Ref. \cite{Miki2021}). Noting that
\begin{align}
    \mathcal{H}_{12}|j_1\rangle|j_2\rangle=V_{j_1j_2}|j_1\rangle|j_2\rangle,~~j_{1(2)}=0_{1(2)},1_{1(2)},
\end{align}
we get
\begin{align}
    \mathcal{U}_{12}(t)|j_1\rangle|j_2\rangle=e^{i\phi_{j_1j_2}t}|j_1\rangle|j_2\rangle,~~j_{1(2)}=0_{1(2)},1_{1(2)},
\end{align}
which shows $~\mathcal{U}_{12}(t)$ is diagonal in the position basis \{$|j_1\rangle|j_2\rangle$\}.

The state of the system at a later time $t$ is now obtained as
\begin{align}
    |\Psi(t)\rangle_{12}&=\mathcal{U}_{12}(t)|\Psi(t=0)\rangle_{12}\nonumber,
\end{align}
which yields Eq. (\ref{time-evolved-two}). The unitary description will be particularly useful in extending the QGEM framework to many-body systems.

%----------------------------------------------------------------------
\section{\label{s3} Three-body QGEM}
%----------------------------------------------------------------------

Let us now consider three test masses $m_1$, $m_2$, and $m_3$ in a general setup [Fig. \ref{qgem_three}]. The initial state is again a tensor product of the spatially superposed states:
\begin{align}\label{three_body_evo}
    |\Psi(t=0)\rangle_{123}=\frac{1}{2\sqrt{2}}\sum_{\substack{j_1=0_1,1_1\\j_2=0_2,1_2\\j_3=0_3,1_3}} |j_1\rangle|j_2\rangle|j_3\rangle.
\end{align}
Let $d_{j_pj_q}$ denote the separation between the states $|j_p\rangle$ and $|j_q\rangle$ for $j_p=0_p,1_p$ and $j_q=0_q,1_q$, where $p\neq q=1,2,3$. As before we assume that $\min_{j_p,j_q} d_{j_pj_q}$ is such that the Casimir-Polder interaction can be neglected.

\begin{figure}[t]
\centering
\includegraphics[width=0.65\textwidth]{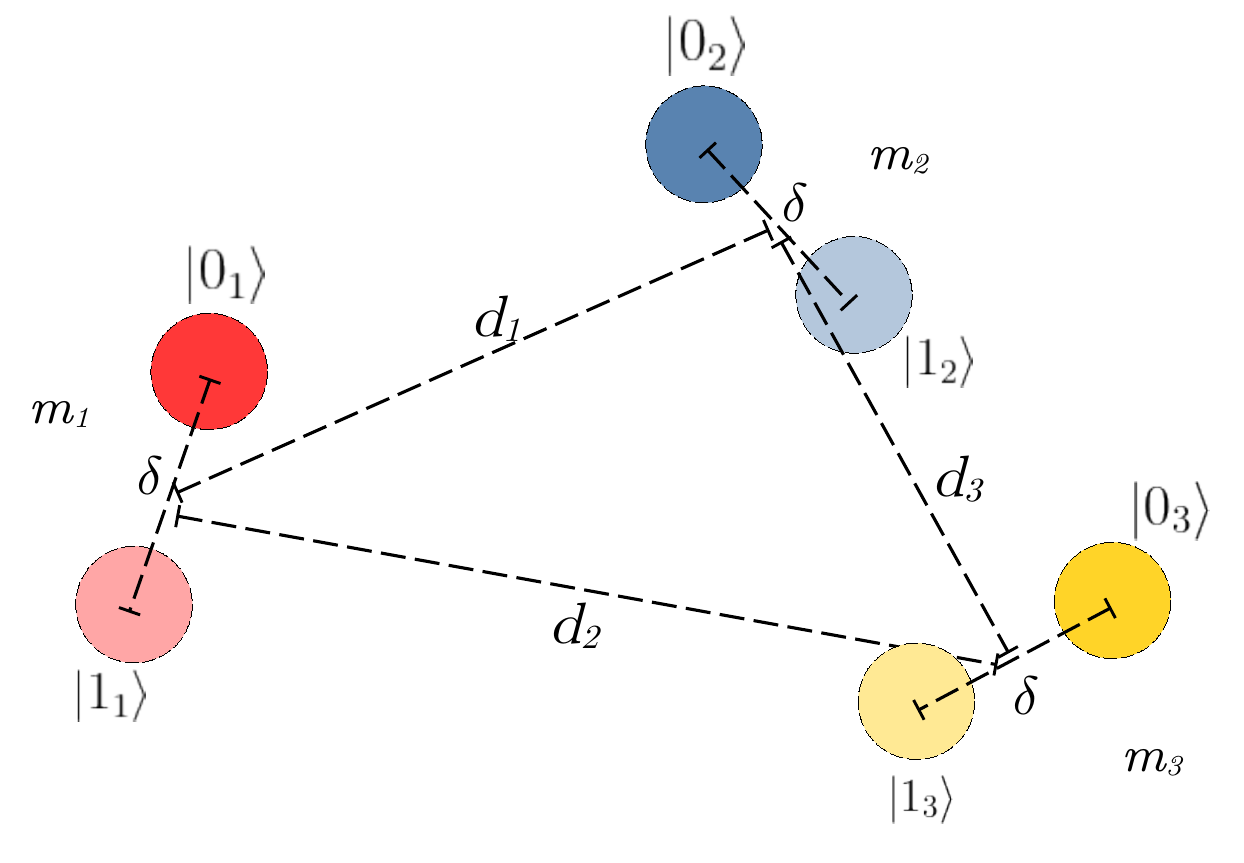}
\caption{\label{qgem_three} Three masses $m_1$, $m_2$ and $m_3$ are each prepared in a superposition of two spatially localised states. The two localised states of mass $m_p$ are denoted as $|0_p\rangle$ and $|1_p\rangle$, where $p=1,2,3$, and they have a separation of $\delta$ between them. The three-body system is arranged in a general setup.}
\end{figure}

The system evolves under mutual gravitational interactions between the test masses according to the equation
\begin{align}\label{3body_time_evolution}
    |\Psi(t)\rangle_{123}=\mathcal{U}_{123}(t)|\Psi(t=0)\rangle_{123},
\end{align}
where $\mathcal{U}_{123}(t)=\exp \left(-\frac{i\mathcal{H}_{123}t}{\hbar}\right)$ is the time-evolution operator and
\begin{align}
    \mathcal{H}_{123}=\mathcal{H}_{12}\otimes\mathcal{I}_{3}+\mathcal{H}_{23}\otimes\mathcal{I}_{1}+\mathcal{H}_{13}\otimes\mathcal{I}_{2}
\end{align}
is the interaction Hamiltonian.
\iffalse
To evaluate the right-hand-side of (\ref{3body_time_evolution}), first we note that
\begin{align}\label{Hamilton}
    \mathcal{H}_{pq}\otimes\mathcal{I}_{r} |j_1\rangle|j_2\rangle|j_3\rangle=V_{j_pj_q}|j_1\rangle|j_2\rangle|j_3\rangle
\end{align}
for $p\neq q\neq r\in\{1,2,3\}$, where $V_{j_pj_q}=-\frac{Gm_pm_q}{d_{j_pj_q}}$. Eq. (\ref{Hamilton}) implies that the operators $\{\mathcal{H}_{pq}\otimes \mathcal{I}_r\}$ are mutually commuting. That means $\mathcal{U}_{123}(t)$ can be decomposed as a product of three mutually commuting unitary operators $~\mathcal{U}_{12}(t)\otimes \mathcal{I}_3$, $~\mathcal{U}_{23}(t)\otimes \mathcal{I}_1$, and $~\mathcal{U}_{13}(t)\otimes \mathcal{I}_2$, where $~\mathcal{U}_{pq}(t)=\exp \left(-\frac{i\mathcal{H}_{pq}t}{\hbar}\right)$. These unitary operators act as
\begin{align}
    \left[\mathcal{U}_{pq}(t)\otimes \mathcal{I}_r\right]|j_1\rangle|j_2\rangle|j_3\rangle=e^{i\phi_{j_pj_q}t}|j_1\rangle|j_2\rangle|j_3\rangle,
\end{align}
where $\phi_{j_pj_q}=-V_{j_pj_q}/\hbar$. 
\fi
Now it is easy to see that the action of $\mathcal{U}_{123}(t)$ on the product basis $\ket{j_1}\ket{j_2}\ket{j_3}$, $j_i\in \lbrace 0_i,1_i\rbrace$, can be written as
\begin{align}\label{3body_unitary}
    \mathcal{U}_{123}(t)|j_1\rangle|j_2\rangle|j_3\rangle=e^{i\phi_{j_1j_2j_3}t}|j_1\rangle|j_2\rangle|j_3\rangle,
\end{align}
where $\phi_{j_1j_2j_3}=\phi_{j_1j_2}+\phi_{j_2j_3}+\phi_{j_1j_3}$. Applying (\ref{3body_unitary}) to (\ref{3body_time_evolution}) and (\ref{three_body_evo}), we therefore obtain
\begin{align}\label{3body_time_evolved_state}
    |\Psi(t)\rangle_{123}=\frac{1}{2\sqrt{2}}\sum_{\substack{j_1=0_1,1_1\\j_2=0_2,1_2\\j_3=0_3,1_3}} e^{i(\phi_{j_1j_2}t +\phi_{j_2j_3}t+\phi_{j_1j_3}t)}|j_1\rangle|j_2\rangle|j_3\rangle.
\end{align}
%\textcolor{blue}{(Comment: What I think that we can condense this part. We write (3.3) as it is and then after some words we directly write (3.7) while mentioning that the local phases can be decomposed in three way. What do you think?)}

It is important to note a key difference in our qubit position basis compared to the standard basis used for qubits like spin-$\frac{1}{2}$ systems. In the case of spin-$\frac{1}{2}$ systems, once the direction of the $z$-axis is fixed, the interpretation of $\ket{\uparrow}$ ($\ket{0}$) and $\ket{\downarrow}$ ($\ket{1}$) is same for all qubits. Any qubit in the state $\dfrac{\ket{0} \pm \ket{1}}{\sqrt{2}}$ is conventionally considered to be oriented along the $\pm x$-axis. However, in our scenario, the definition of spatial states $\ket{0_p}$ and $\ket{1_p}$ for mass $m_p$ is independent and, in general, different from the definition of $\ket{0_q}$ and $\ket{1_q}$ for mass $m_q$, $p\neq q$. Consequently, the superposition orientation of mass $m_p$ in the state $\dfrac{\ket{0_p} + \ket{1_p}}{\sqrt{2}}$ varies with different values of `$p$'. This implies that the pairwise Hamiltonians being simultaneously diagonalizable in the position basis cannot be straightforwardly interpreted as a ``single choice of basis'', in the usual sense of qubits, makes all pairwise Hamiltonians diagonal.

%-----------------------------------------------
\subsection*{Entanglement properties}
%-----------------------------------------------

The time-evolved state $|\Psi(t)\rangle_{123}$ is a three-qubit pure state. The entanglement properties of such a state can be studied by computing entanglement across all three bipartitions $1|23$, $2|13$, and $3|12$. Observe that in each bipartition the state can be expressed as a bipartite pure state in $\mathbb{C}^2\otimes\mathbb{C}^4$, having Schmidt rank at most 2.\footnote{Any bipartite pure state $|\psi\rangle_{AB}\in\mathbb{C}^{d}\otimes\mathbb{C}^{d'}$ can be written in the Schmidt decomposition form as $|\psi\rangle_{AB}=\sum_{i=1}^k \lambda_i|\phi_i\rangle_A|\eta_i\rangle_B$ where $\{|\phi_i\rangle_A\}$ and $\{|\eta_i\rangle_B\}$ are orthogonal vectors in $\mathbb{C}^d$ and $\mathbb{C}^{d'}$ respectively, and $k=\min\{d,d'\}$. The Schmidt rank of the state, which is the number of nonzero Schmidt coefficients $\lambda_i$'s, can be at most $k=\min \{d,d'\}$.} Therefore, only two of the four dimensions in $\mathbb{C}^4$ are in fact necessary. So for our purpose we can effectively consider the system as a pair of qubits that allows us to compute the concurrence exactly.

To compute concurrence across a bipartition $p|qr$, where $p\neq q\neq r\in \{1, 2, 3\}$, we proceed as follows. First we write $|\Psi(t)\rangle_{123}$ [given by (\ref{3body_time_evolved_state})] as
\begin{align}\label{C2C4}
    |\Psi(t)\rangle_{p|qr}=\frac{1}{\sqrt{2}}(|0_p\rangle|\xi_{qr}\rangle+|1_p\rangle|\zeta_{qr}\rangle),
\end{align}
where
\begin{align*}
    |\xi_{qr}\rangle &= \frac{1}{2} \sum_{\substack{j_q=0_q,1_q\\j_r=0_r,1_r}} e^{i\phi_{0_pj_qj_r}t}|j_q\rangle|j_r\rangle,\\
    |\zeta_{qr}\rangle &= \frac{1}{2} \sum_{\substack{j_q=0_q,1_q\\j_r=0_r,1_r}} e^{i\phi_{1_pj_qj_r}t}|j_q\rangle|j_r\rangle.
\end{align*}
Though equation (\ref{C2C4}) is not in the Schmidt form, it can be expressed in the Schmidt form with at most two terms, which is why only two of the four dimensions associated with the state space $\mathbb{C}^4$ of $qr$ are required. This implies that $|\Psi(t)\rangle_{p|qr}$ can be treated as a pure state in $\mathbb{C}^2\otimes\mathbb{C}^2$, the concurrence of which is given by twice the product of the Schmidt coefficients. The Schmidt coefficients are calculated as
\begin{align}
    \lambda_{\pm}(t)=\sqrt{\frac{1}{2}\pm\vert\eta(t)\vert},
\end{align}
where
\begin{align}\label{eta}
    \eta(t)=\frac{1}{8} \sum_{\substack{j_q=0_q,1_q\\j_r=0_r,1_r}} e^{i(\phi_{0_pj_qj_r}-\phi_{1_pj_qj_r})t}.
\end{align}
The concurrence is therefore given by
\begin{align}\label{3concurrence}
    \mathbf{C}_{p|qr}(t)=2\lambda_+\lambda_-=2\sqrt{\frac{1}{4}-\vert\eta(t)\vert^2}.
\end{align}
We will now express $\vert\eta(t)\vert$ as a function of the entangling phases. First we expand (\ref{eta}) as
\begin{align}
    \eta(t)=\frac{1}{8}\left(\sum_{j_q=0_q,1_q} e^{i(\phi_{0_pj_q}-\phi_{1_pj_q})t}\right)\left(\sum_{j_r=0_r,1_r} e^{i(\phi_{0_pj_r}-\phi_{1_pj_r})t}\right).\nonumber
\end{align}
Further simplification gives
\begin{align}\label{eta2}
    \eta(t)=\frac{1}{8}e^{i(\phi_{0_p0_q}- \phi_{1_p0_q}+\phi_{0_p0_r}-\phi_{1_p0_r})t}\left(1+e^{i\Phi_{pq}t}\right)\left(1+e^{i\Phi_{pr}t}\right),
\end{align}
where $\Phi_{ab}$ is the entangling phase associated with qubit-pair $(a,b)$ and defined as
\begin{align}\label{ent}
    \Phi_{ab}=\vert\phi_{0_a1_b}+\phi_{1_a0_b}-\phi_{0_a0_b}-\phi_{1_a1_b}\vert.
\end{align}
From (\ref{eta2}) it follows that
\begin{align}\label{verteta}
    \vert\eta(t)\vert=\frac{1}{8}\vert1+e^{i\Phi_{pq}t}\vert\vert1+e^{i\Phi_{pr}t}\vert
\end{align}
Plugging (\ref{verteta}) in (\ref{3concurrence}) and simplifying we get
\begin{align}\label{3icon}
    \mathbf{C}_{p|qr}(t)=\sqrt{1-\cos^2\left(\frac{\Phi_{pq}t}{2}\right)\cos^2\left(\frac{\Phi_{pr}t}{2}\right)},
\end{align}
Observe that $\mathbf{C}_{p|qr}(t)$ depends only on $\Phi_{pq}$ and $\Phi_{pr}$ associated with mass-pairs $(m_p,m_q)$ and $(m_p,m_r)$.

\paragraph*{Properties of $\mathbf{C}_{p|qr}(t)$:}
\begin{enumerate}
    \item \label{prop1} $\mathbf{C}_{p|qr}(t)$ vanishes whenever both $\cos\left(\frac{\Phi_{pq}t}{2}\right)=\pm 1$ and $\cos\left(\frac{\Phi_{pr}t}{2}\right)=\pm 1$; that is, when
    \begin{align}
        \Phi_{pq}t=2n_{pq}\pi~~\text{and}~~\Phi_{pr}t=2n_{pr}\pi,~~~n_{pq},n_{pr}=0,1,2,\ldots
    \end{align}
    are satisfied simultaneously, which happens when $\frac{\Phi_{pq}}{\Phi_{pr}}$ is a rational number. Note that $\mathbf{C}_{p|qr}(t)=0$ implies $|\Psi(t)\rangle_{p|qr}$ is a product state of the form $|\chi(t)\rangle_p\otimes|\chi'(t)\rangle_{qr}$.

    \item A corollary of the first property is that $\mathbf{C}_{p|qr}(t)>0$ for all $t > 0$ if $~\frac{\Phi_{pq}}{\Phi_{pr}}$ is not a rational number.

    \item $\mathbf{C}_{p|qr}(t)=1$ whenever $\cos\left(\frac{\Phi_{pq}t}{2}\right)=0$ or/and $\cos\left(\frac{\Phi_{pr}t}{2}\right)=0$; that is, when
    \begin{align}
        \Phi_{pq}t=2(n_{pq}+1)\pi~~\text{or/and}~~\Phi_{pr}t=2(n_{pr}+1)\pi,~~~n_{pq},n_{pr}=0,1,2,\ldots
    \end{align}
    Note that $\mathbf{C}_{p|qr}(t)=1$ implies $|\Psi(t)\rangle_{p|qr}$ is maximally entangled across the subspace $\mathbb{C}^2$ of $p$ and the subspace $\mathbb{C}^4$ of $qr$.

    \item $\mathbf{C}_{p|qr}(t)$ exhibits periodic behaviour if
    \begin{align}
        \frac{\Phi_{pq}}{\Phi_{pr}}=\frac{n_{pq}}{n_{pr}}
    \end{align}
    with time period $T=\frac{2n_{pq}\pi}{\Phi_{pq}}=\frac{2n_{pr}\pi}{\Phi_{pr}}$.
    
    \end{enumerate}

Based on the above properties we have the following observations on the three-body entanglement.
\begin{proposition}[Separability]\label{Separability} The system becomes fully separable when the conditions
\begin{align}
    \Phi_{12}t=2n_{12}\pi,\Phi_{23}t=2n_{23}\pi,~\text{and}~\Phi_{13}t=2n_{13}\pi,~~n_{12},n_{23},n_{13}=0,1,2,\ldots
\end{align}
are satisfied simultaneouly. Equivalently, entanglement across all bipartitions will vanish periodically provided the ratios of the entangling phases are rational numbers.
\end{proposition}

The proof follows from the first property of $\mathbf{C}_{p|qr}(t)$ applied to each bipartition.

\begin{proposition}[Genuine three-body entanglement] \label{prop:3GME} The time-evolution leads to genuine three-body entanglement if and only if at least two of the entangling phases are nonzero.
\end{proposition}

\begin{proof}
    First note that the state would be genuinely entangled if and only if entanglement is nonzero across all bipartitions. Now, if any two of the three entangling phases are zero, then there exists a bipartition where there is no entanglement at any time $t>0$. So to have nonzero entanglement across every bipartition no more than one entangling phase can be zero.
\end{proof}

It is clear that the time evolution in general will lead to genuine three-body entanglement. However, at a later time $t$ the state could become fully separable [see, Proposition \ref{Separability}]. Then the question is whether this is always the case for any setup. The following proposition answers this question in negative.

\begin{proposition}[Sustainability]
    The state (\ref{3body_time_evolved_state}) is genuine three-body entangled at all times $t > 0$ provided all three entangling phases are nonzero and the ratio of any two entangling phases is not rational.
\end{proposition}

\begin{proof}
    Suppose that one of the entangling phases, say $\Phi_{12}$ is zero. Then we have
    \begin{align}
    \mathbf{C}_{1|23}&=\sqrt{1-\cos^2\left(\frac{\Phi_{13}t}{2}\right)}, \nonumber\\
    \mathbf{C}_{2|13}&=\sqrt{1-\cos^2\left(\frac{\Phi_{23}t}{2}\right)}. \nonumber
    \end{align}
    Here, the system will exhibit genuine entanglement, where entanglement across the bipartitions $1|23$ and $2|13$ will vanish at times $t=\frac{2n_{13}\pi}{\Phi_{13}}$ and $t=\frac{2n_{23}\pi}{\Phi_{23}}$ respectively, for $n_{13}, n_{23} = 0,1,2,\ldots$, and thereby, the system becoming biseparable at those instants. So none of the entangling phases can be zero. But clearly this is not enough. In addition, we require the second property of $\mathbf{C}_{p|qr}(t)$ to hold for all bipartitions. This completes the proof.
\end{proof}

We now show that provided certain conditions are being met the time evolution leads to a maximally entangled GHZ state.

\begin{proposition}[GHZ entanglement]
    The state (\ref{3body_time_evolved_state}) periodically evolves into a maximally entangled GHZ state if any two entangling phases are nonzero and their ratio is that of two odd integers.
\end{proposition}

\begin{proof}
    Let any two entangling phases, say $\Phi_{12}$ and $\Phi_{13}$ be nonzero and $\frac{\Phi_{12}}{\Phi_{13}}=\frac{2n_{12}+1}{2n_{13}+1}$. Therefore, we can write $\Phi_{12}=(2n_{12}+1)\Phi$ and $\Phi_{13}=(2n_{13}+1)\Phi$ for some $\Phi\neq0$. Then we have
    \begin{align}
        \mathbf{C}_{1|23}(t) &= \sqrt{1-\cos^2\left[\frac{(2n_{12}+1)\Phi t}{2}\right]\cos^2\left[\frac{(2n_{13}+1)\Phi t}{2}\right]},\nonumber \\
        \mathbf{C}_{2|13}(t) &= \sqrt{1-\cos^2\left[\frac{(2n_{12}+1)\Phi t}{2}\right]\cos^2\left[\frac{\Phi_{23} t}{2}\right]},\nonumber \\
        \mathbf{C}_{3|12}(t) &= \sqrt{1-\cos^2\left[\frac{\Phi_{23} t}{2}\right]\cos^2\left[\frac{(2n_{13}+1)\Phi t}{2}\right]}.\nonumber
    \end{align}
    Observe that $\mathbf{C}_{1|23}(t)=\mathbf{C}_{2|13}(t)=\mathbf{C}_{3|12}(t)=1$ whenever $\Phi t=(2n+1)\pi$ for $n=0,1,2,\ldots$, which follows from the fact that the product of any two odd integers is also an odd integer. This is also the case even if we let $\Phi_{23}=0$. Therefore, the state (\ref{3body_time_evolved_state}) periodically evolves into a maximally entangled GHZ state with 1 ebit of entanglement across every $\mathbb{C}^2\otimes\mathbb{C}^4$ bipartition.
\end{proof}

%---------------------------------------
\subsection*{Pairwise entanglement and three-body correlations}
%---------------------------------------
A fundamental property of many-body entanglement is that entanglement is monogamous \cite{CKW2000}. Specifically, the distribution of bipartite entanglement in a three-qubit system in a pure state $|\psi\rangle_{123}$ satisfies an equality of the form
%---------------------------------------------
\begin{align}\label{monogamy}
    \tau_{ab}+\tau_{ac}+\tau_{123}=\tau_{a|bc},
\end{align}
%---------------------------------------------
for all $a\neq b\neq c\in\{1,2,3\}$, where $\tau_{a|bc} = \mathbf{C}^2_{a|bc}$, $\tau_{xy}=\mathbf{C}_{xy}^2$ is the tangle where $\mathbf{C}_{xy}=\mathbf{C}(\rho_{xy})$ is the concurrence of the two-qubit reduced density matrix $\rho_{xy}=\Tr_z \left(|\psi\rangle_{xyz}\langle\psi|\right)$, and $\tau_{123}$ is the 3-tangle which is a measure of three-qubit correlations. The monogamy equalities tell us that entanglement cannot be freely shared.

We would like to obtain exact expressions of two-qubit concurrences for they quantify pairwise entanglement and the 3-tangle in our case. The motivation of computing 3-tangle comes from the fact that whenever the state (\ref{3body_time_evolved_state}) is genuinely entangled, a nonzero (zero) value of the 3-tangle would indicate it belongs to the GHZ (W) class \cite{CKW2000}.

The 3-tangle $\tau_{123}$ of (\ref{3body_time_evolved_state}) can be exactly calculated.\footnote{The calculation is tedious but straightforward. We start from the definition \cite{CKW2000}, write down all the terms and combine them to obtain the $f_i$s that are finally expressed as a function of the entangling phases.} The explicit formula is given by
\begin{align}\label{3tangle}
    \tau_{123}(t)=\frac{1}{16}\vert f_1(t)-2f_2(t)+8f_3(t)\vert,
\end{align}
where
\begin{align}
    f_1(t) &= 1+\sum_{p\neq q\neq r} e^{2i(\Phi_{pq}+\Phi_{pr})t},\nonumber\\
    f_2(t) &= \sum_{p\neq q\neq r} e^{i(\Phi_{pq}+\Phi_{pr})t} + \sum_{p\neq q\neq r} e^{i(2\Phi_{pq}+\Phi_{pr}+\Phi_{qr})t},\nonumber\\
    f_3(t) &= \sum_{p<q} e^{i\Phi_{pq}t},\nonumber
\end{align}
where $p,q,r\in\{1,2,3\}$.

Now from (\ref{3tangle}) one finds that
\begin{align}\label{2reduced_concurrence}
    \mathbf{C}_{pq}(t)=\frac{1}{\sqrt{2}}\sqrt{\tau_{p|qr}+\tau_{q|pr}-\tau_{r|pq}-\tau_{123}},
\end{align}
which can be exactly computed using (\ref{3icon}) and (\ref{3tangle}).

So now we have the expressions of all the relevant quantities we require to study the entanglement properties of the time-evolved state $|\Psi(t)\rangle_{123}$. In particular, once we specify a setup, linear, parallel, or otherwise, we could calculate everything that we need to know regarding entanglement and three-body correlations of the system.

%----------------------------------------------------------------------
\section{\label{s4} $N$-body QGEM}
%----------------------------------------------------------------------
Next we extend the treatment to a system of $N$ masses $m_1,m_2,\cdots,m_N$, each initially prepared in a superposed state $|\psi\rangle_p=\frac{1}{\sqrt{2}}\left(|0_p\rangle+|1_p\rangle \right)$ for $p=1,2,\ldots,N$. As before, we do not fix any specific geometry or setup. The initial state of the system is given by
\begin{align}
    |\Psi(t=0)\rangle_{12\ldots N} &= \bigotimes_{p=1}^N |\psi\rangle_p = \frac{1}{2^{N/2}}\sum_{j_1=0_1,1_1} \sum_{j_2=0_2,1_2}\cdots\sum_{j_N=0_N,1_N} |j_1\rangle|j_2\rangle\ldots|j_N\rangle.
\end{align}

The time-evolution of the $N$-body state under the mutual gravitational interactions between the masses is described by
\begin{align}\label{Nbody_time_evolution}
    |\Psi(t)\rangle_{12\ldots N} &= \mathcal{U}_{12\ldots N}(t)~|\Psi(t=0)\rangle_{12\ldots N}\nonumber\\
    &=\exp\left(-\frac{i\mathcal{H}_{12\ldots N}t}{\hbar}\right)|\Psi(t=0)\rangle_{12\ldots N},
\end{align}
where $\mathcal{U}_{12\ldots N}(t)$ is the unitary time-evolution operator generated by the $N$-body interaction Hamiltonian
\begin{align}
    \mathcal{H}_{12\ldots N}=\sum_{p<q=1}^N \mathcal{H}_{pq}\otimes\mathcal{I}.
\end{align}
Here, $\mathcal{H}_{pq}$ is the two-body interaction Hamiltonian (\ref{Hamiltonian}) corresponding to the mass-pair $(m_p,m_q)$, $p,q\in\{1,2,\ldots,N\}$ and $\mathcal{I}$ is the identity operator acting on the tensor product of Hilbert spaces of all the masses except $(m_p,m_q)$.

The operators $\left\{\mathcal{H}_{pq}\otimes\mathcal{I}\right\}$ are mutually commuting as evident from
\begin{align}
    \left(\mathcal{H}_{pq}\otimes\mathcal{I}\right)|j_1\rangle|j_2\rangle\ldots|j_N\rangle =V_{j_pj_q}|j_1\rangle|j_2\rangle\ldots|j_N\rangle,
\end{align}
where $V_{j_pj_q}=-\frac{Gm_pm_q}{d_{j_pj_q}}$, and $d_{j_pj_q}$ is the distance between the states $|j_p\rangle$ and $|j_q\rangle$. Consequently, $\mathcal{U}_{12\ldots N}(t)$ can be decomposed as a product of ${N\choose 2}$ mutually commuting unitary operators as
\begin{align}
    \mathcal{U}_{12\ldots N}(t)=\prod_{p<q=1}^N \left[\mathcal{U}_{pq}(t)\otimes\mathcal{I}\right]=\prod_{p<q=1}^N \exp\left(-\frac{i\mathcal{H}_{pq}t}{\hbar}\right)\otimes\mathcal{I}.
\end{align}
The action of $\mathcal{U}_{12\ldots N}(t)$ is therefore given by 
\begin{align}\label{uni}
    \mathcal{U}_{12\ldots N}(t)~|j_1\rangle|j_2\rangle\ldots|j_N\rangle=e^{i\phi_{j_1j_2\ldots j_N}t}|j_1\rangle|j_2\rangle\ldots|j_N\rangle~~\text{for}~j_p=0_p,1_p,~p=1,2,\ldots,N,
\end{align}
where
\begin{align}\label{bw}
\phi_{j_1j_2\ldots j_N}=\sum_{p<q=1}^N \phi_{j_pj_q}    
\end{align}
and $\phi_{j_pj_q}=-\frac{V_{j_pj_q}}{\hbar}$. Putting (\ref{uni}) in (\ref{Nbody_time_evolution}), we get the expression for the time-evolved $N$-body state as
\begin{align}\label{N-body-time-evolved-state}
     |\Psi(t)\rangle_{12\ldots N} = \frac{1}{2^{N/2}}\sum_{j_1=0_1,1_1} \sum_{j_2=0_2,1_2}\cdots\sum_{j_N=0_N,1_N} e^{i\phi_{j_1j_2\ldots j_N}t}|j_1\rangle|j_2\rangle\ldots|j_N\rangle.
\end{align}

%---------------------------------------
\subsection*{Entanglement properties of the time-evolved $N$-body state}
%---------------------------------------
To explore the entanglement properties of the $N$-qubit pure state $|\Psi(t)\rangle_{12\ldots N}$, we divide the system into two parts: one containing $k$ masses, where $1\leq k \leq \left\lfloor\frac{N}{2}\right\rfloor$, and the other containing $(N-k)$ masses. Note that for each $1\leq k < \frac{N}{2}$, there are ${N\choose k}$ different possible bipartitions, which are obtained through various combinations of the $k$ masses. In the case where $k=\frac{N}{2}$ (applicable when $N$ is even), there are $\frac{1}{2}{N\choose N/2}$ such different bipartitions. The I-concurrence measuring the bipartite entanglement \cite{Rungta2001} across any such bipartition can be computed exactly. 

\begin{theorem}
The I-concurrence of the state (\ref{N-body-time-evolved-state}) for any $k|(N-k)$ bipartition ($1\leq k\leq \left\lfloor\frac{N}{2}\right\rfloor$) is given by
\begin{align}\label{N-icon}
    \mathcal{C}_{p_1p_2\ldots p_k|p_{k+1}\ldots p_N}(t)=\sqrt{\frac{2^k-1}{2^{k-1}}-2^{2-k}\Lambda(t)},
\end{align}
where
\begin{small}
\begin{align}\label{series}
    \Lambda(t)=&\frac{1}{2} \sum_{a=p_1}^{p_k} \left[\prod_{b=p_{k+1}}^{p_N} \cos^2\left(\frac{\Phi_{ab}}{2}t \right)\right]\nonumber\\
    &+\frac{1}{2^2}\sum_{s_1=0}^{1}\sum_{\substack{a_1<a_2,\\a_1,a_2={p_1}}}^{p_k} \left[\prod_{b=p_{k+1}}^{p_N} \cos^2\left(\frac{\Phi_{a_1b}+(-1)^{s_1}\Phi_{a_2b}}{2}t \right)\right]\nonumber\\
    &+\frac{1}{2^3}\sum_{s_1,s_2=0}^1\sum_{\substack{a_1<a_2<a_3\\a_1,a_2,a_3={p_1}}}^{p_k} \left[\prod_{b=p_{k+1}}^{p_N} \cos^2\left(\frac{\Phi_{a_1b}+(-1)^{s_1}\Phi_{a_2b}+(-1)^{s_2}\Phi_{a_3b}}{2}t\right)\right]+\ldots\nonumber\\
    &+\frac{1}{2^k}\sum_{s_1,\ldots ,s_{(k-1)}=0}^1\sum_{\substack{a_1<a_2<\ldots<a_k\\a_1,a_2,\ldots,a_k=p_1}}^{p_k} \left[\prod_{b=p_{k+1}}^{p_N} \cos^2\left(\frac{\Phi_{a_1b}+(-1)^{s_1}\Phi_{a_2b}+\ldots+(-1)^{s_{(k-1)}}\Phi_{a_kb}}{2}t \right)\right],
\end{align}
\end{small}
for $p_1<p_2\ldots<p_k \in \{1,2,\ldots,N\}$ and $p_{k+1}<p_{k+2}\ldots<p_N \in \{1,2,\ldots,N\}\backslash\{p_1,p_2,\ldots,p_k\}$.
\end{theorem}

The proof is provided in Appendix \ref{app}. Note that there are total ${N \choose 2}$ entangling phases 
\begin{align*}
    \Phi_{xy}=\vert\phi_{0_x1_y}+\phi_{1_x0_y}-\phi_{0_x0_y}-\phi_{1_x1_y}\vert~\text{for}~x<y=1,2,\ldots,N,
\end{align*}
each associated with a mass-pair. These entangling phases depend on the specific geometry and setup of the system. However, the I-concurrence for a given bipartition depends only on the entangling phases of those mass-pairs whose masses are in the alternative parts of the bipartition.

\begin{corollary}
    The I-concurrence for any $1|(N-1)$ (one-vs-rest) bipartition is given by
\begin{align}\label{one-vs-rest}
    \mathcal{C}_{p_1|p_2\ldots p_N}(t)= \sqrt{1-\prod_{b=p_2}^{p_N}\cos^2\left(\frac{\Phi_{p_1b}t}{2}\right)},
\end{align}
where $p_1\in\{1,2,\ldots,N\}$ and $p_2<p_3\ldots<p_N\in\{1,2,\ldots,N\}\backslash\{p_1\}$.
\end{corollary}

With the exact expressions of I-concurrence for every bipartition at our disposal, we can study entanglement distribution in the system for any specified geometry and setup.   

\begin{figure}[t]
\centering
\includegraphics[width=0.75\textwidth]{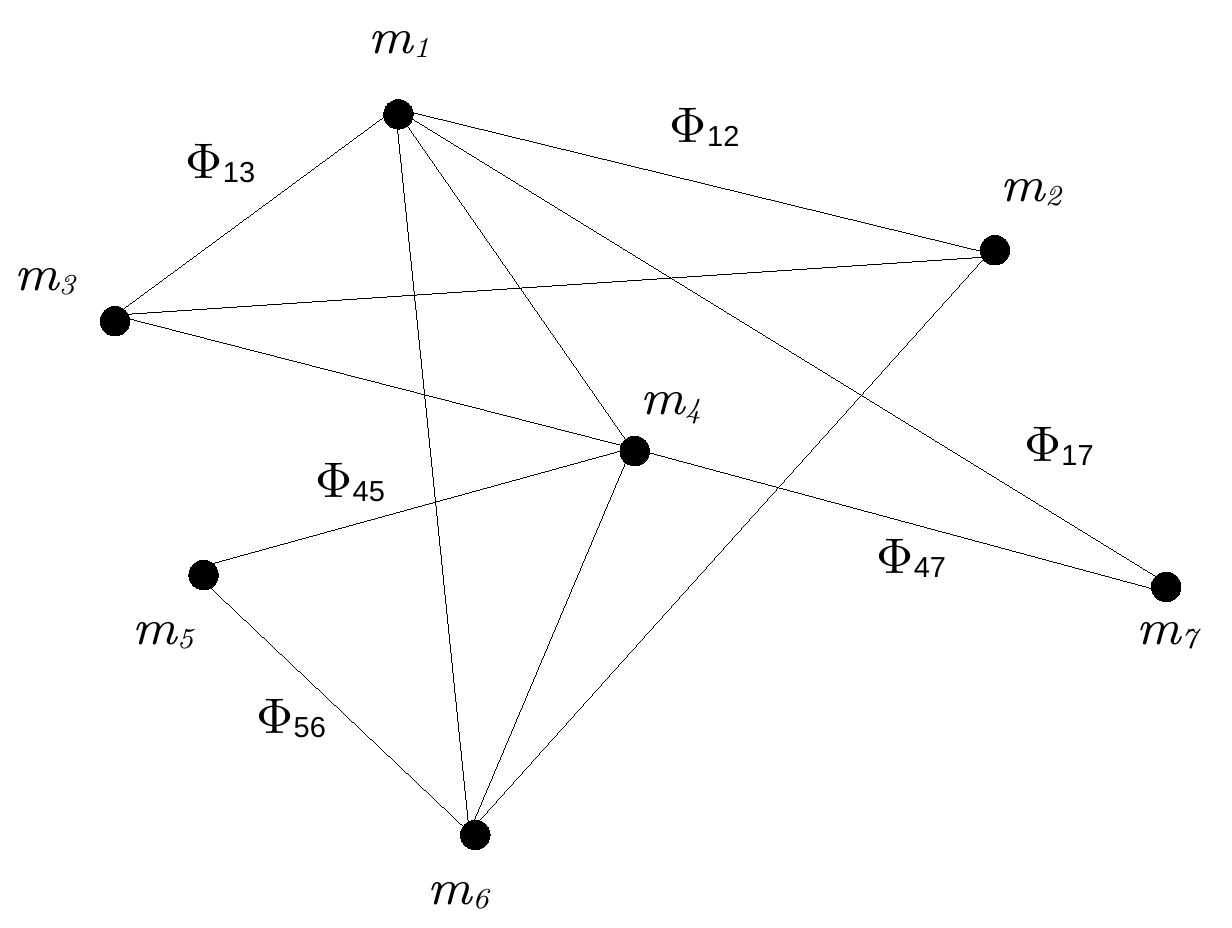}
\caption{\label{qgem_graph} Representation of $N$-body system as a graph. We have considered $N=7$ as an example. The masses $\{m_p\}_{p=1}^7$ form the nodes of the graph. The edge that connects two nodes represents the entangling phase of the corresponding mass-pair. Absence of any edge between two particular nodes implies that the corresponding entangling phase for the mass-pair is 0. For example, $\Phi_{67}=0$ or $\Phi_{25}=0$. }
\end{figure}

\subsubsection*{Graphical representation}
We now represent the $N$-body system as a graph. This representation is very useful as it provides a practical way to calculate the I-concurrence for a given bipartition pictorially. Moreover, it offers a natural framework for discussing gravity-induced many-body entanglement properties.

Consider the system of $N$ masses as a graph $G(V,E)$ \cite{Diestel2017}. Here, $V=\{m_1,m_2,\ldots,m_N\}$ is the set of vertices (or nodes) representing the masses, and $E\subseteq \{\{m_p,m_q\}|m_p,m_q\in V \land m_p\neq m_q\}$ is the set of edges. Each edge connects two nodes and represents the entangling phase associated with the corresponding mass-pair. For example, the edge $\{m_p,m_q\}$ denotes the entangling phase $\Phi_{pq}$ of the mass-pair $(m_p,m_q)$ [see Fig. \ref{qgem_graph}]. The absence of an edge between two nodes implies zero entangling phase for the corresponding mass-pair. We can create a $k|(N-k)$ bipartition of the graph by selecting a set of $k$ nodes $V_k=\{m_{p_1},m_{p_2},\ldots,m_{p_k}\}\subset V$ as one part and the remaining set of $(N-k)$ nodes $V'_k=V-V_k$ as the other part. The set of edges, whose one end-vertex is in $V_k$ and the other end-vertex is in $V'_k$, is denoted as $E(V_k,V'_k)=\{\{m_a,m_b\}|m_a\in V_k,m_b\in V'_k\}\subset E$ \cite{Diestel2017}. Clearly, the I-concurrence (\ref{N-icon}) for this bipartition depends solely on the edges in $E(V_k,V'_k)$.

Let us now demonstrate how this representation can be used to calculate the I-concurrence for a given bipartition, with the help of a specific example: Consider a 6-body system, and we are to calculate the I-concurrence for the $(1,2,5)|(3,4,6)$ bipartition. Accordingly, we define $V_3=\{m_1,m_2,m_5\}$ and $V'_3=\{m_3,m_4,m_6\}$. Now, note that for a given $k|(N-k)$ bipartition, there are $k$ terms in (\ref{series}). Thus, in this case, we need to calculate $3$ terms to obtain $\Lambda(t)$. 

\begin{figure}[t]
    \centering
    \includegraphics[width=\textwidth]{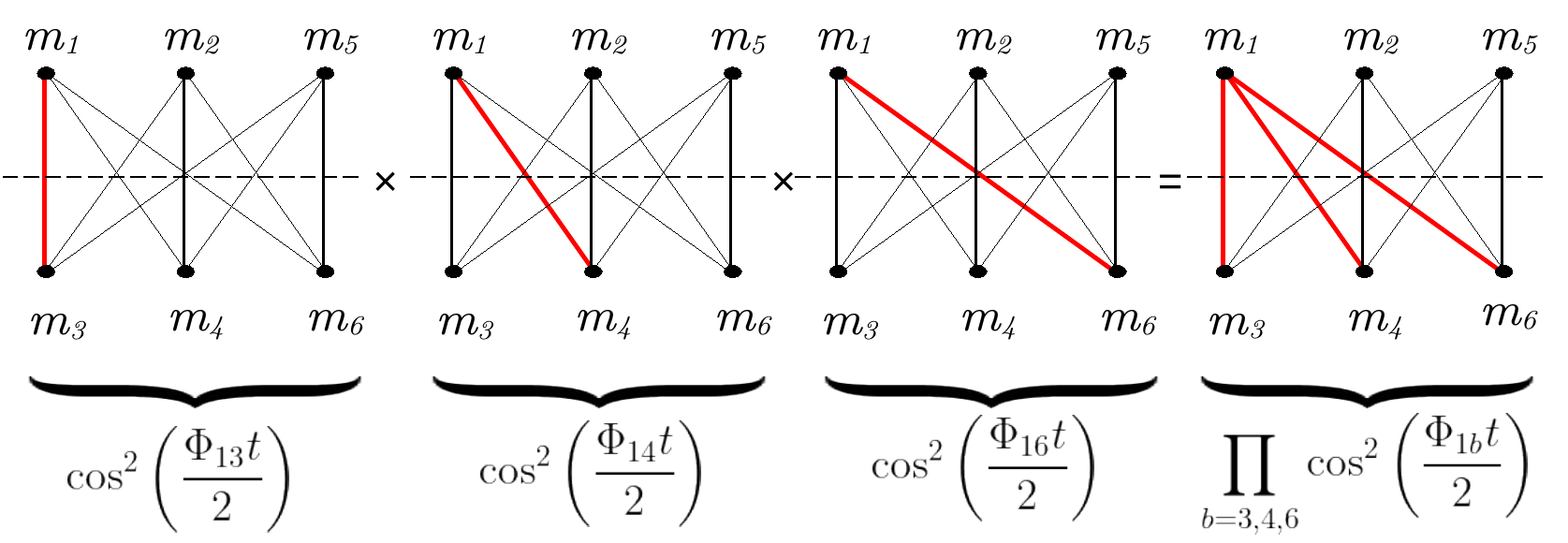}
    \includegraphics[width=0.85\textwidth]{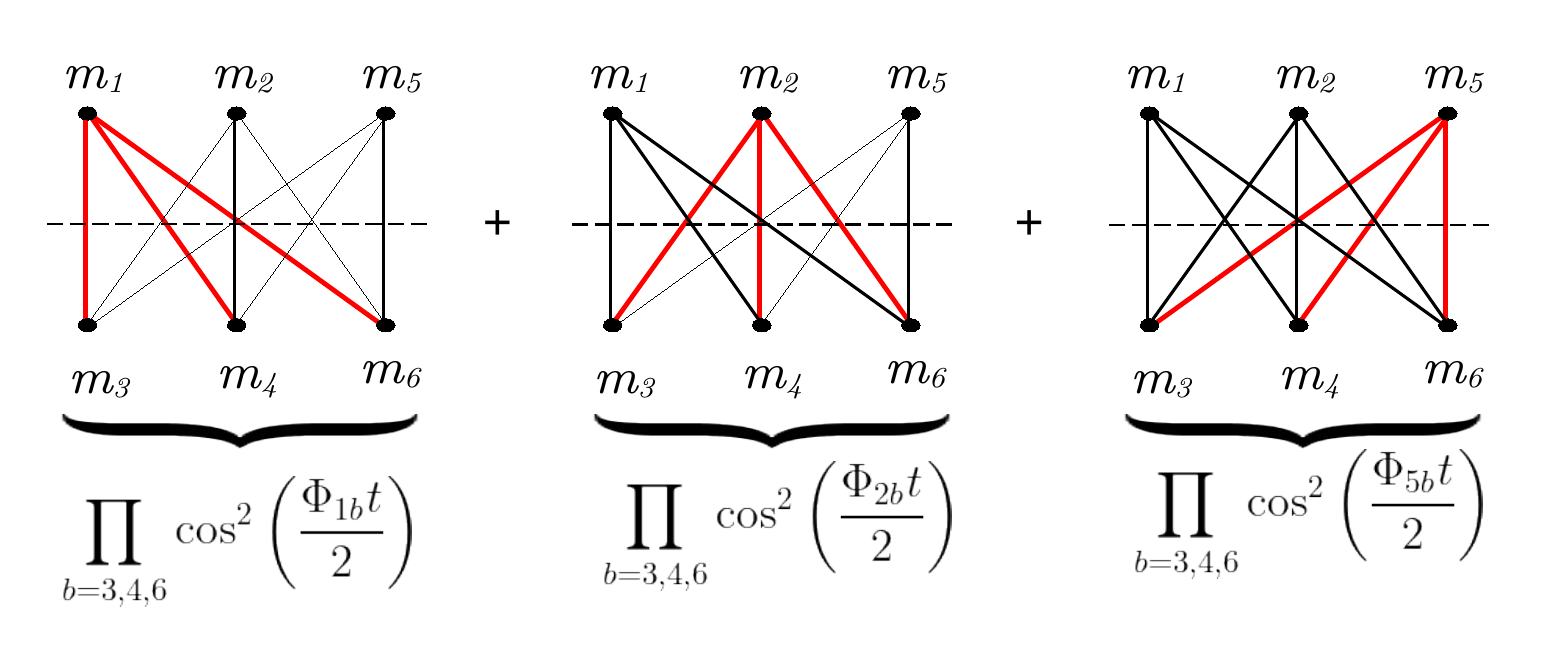}
    \caption{Diagram for calculating the first term in the series for calculating the I-concurrence for the $(1,2,5)|(3,4,6)$ bipartition}
    \label{N-icon-1st}
\end{figure}

\noindent\textit{First term:} Consider the set of edges $E(\{m_1\},V'_3)=\{\{m_1,m_3\},\{m_1,m_4\},\{m_1,m_6\}\}\subset E(V_3,V'_3)$. With each edge $\{m_1,m_b\}\in E(\{m_1\},V'_3)$, associate the function $\cos^2\left(\frac{\Phi_{1b}}{2}t\right)$, for $b=3,4,6$. Take the product of these three functions [see Fig. \ref{N-icon-1st}]. Repeat this process for the sets $E(\{m_2\},V'_3)$ and $E(\{m_5\},V'_3)$. The first term, denoted as $T_1$, is obtained by summing these three quantities (products) with a multiplicative factor of $\frac{1}{2}$:
\begin{align}
    T_1=\frac{1}{2}&\left\lbrace\cos^2 \left(\frac{\Phi_{13}}{2}t\right)\cos^2 \left(\frac{\Phi_{14}}{2}t\right)\cos^2 \left(\frac{\Phi_{16}}{2}t\right)+\cos^2 \left(\frac{\Phi_{23}}{2}t\right)\cos^2 \left(\frac{\Phi_{24}}{2}t\right)\cos^2 \left(\frac{\Phi_{26}}{2}t\right)\right.\nonumber\\
    &+\left.\cos^2 \left(\frac{\Phi_{35}}{2}t\right)\cos^2 \left(\frac{\Phi_{45}}{2}t\right)\cos^2 \left(\frac{\Phi_{56}}{2}t\right)\right\rbrace.
\end{align}

\begin{figure}[t]
    \centering
    \includegraphics[width=\textwidth]{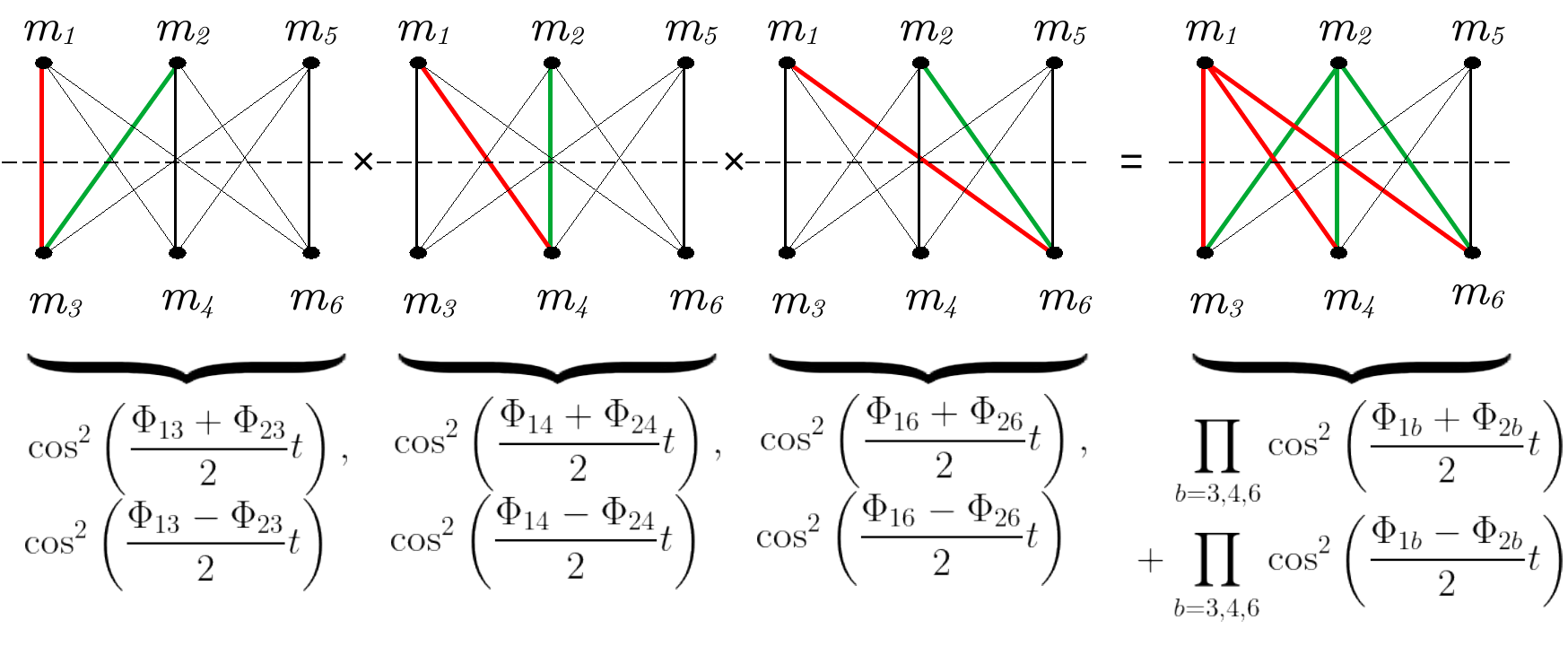} 
    \includegraphics[width=0.85\textwidth]{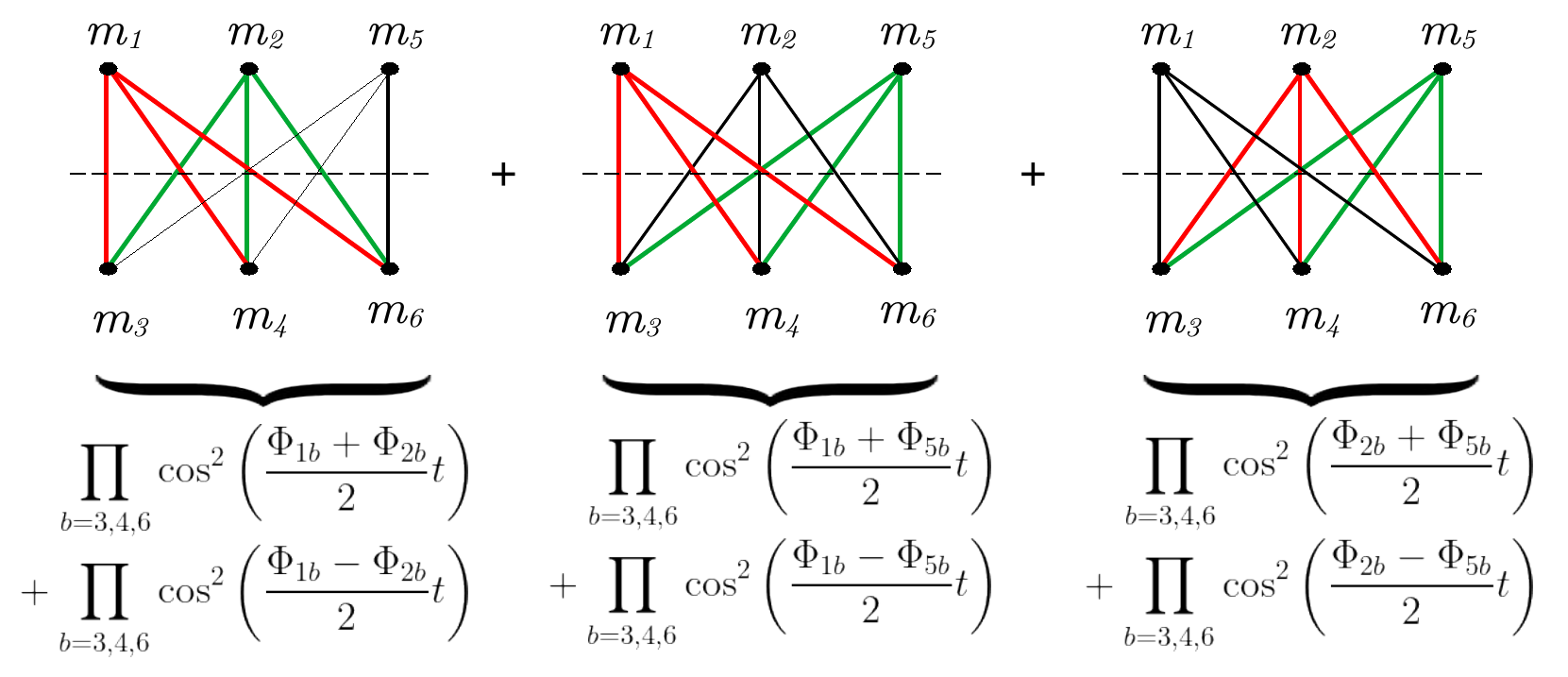}
    \caption{Diagram for calculating the second term in the series for calculating the I-concurrence for the $(1,2,5)|(3,4,6)$ bipartition}
    \label{N-icon-2nd}
\end{figure}

\noindent\textit{Second term:} Now, consider the set of edges $E(\{m_1,m_2\},\{m_b\})$ where $m_b\in V'_3$. Associate two functions with this set: $f_{b,+}=\cos^2\left(\frac{\Phi_{1b}+\Phi_{2b}}{2}t\right)$ and $f_{b,-}=\cos^2\left(\frac{\Phi_{1b}-\Phi_{2b}}{2}t\right)$. Take the product of $f_{b,+}$ for different values of $b=3,4,6$, and do the same for $f_{b,-}$ separately. Then, add the two products: $f_{3,+}f_{4,+}f_{6,+}+f_{3,-}f_{4,-}f_{6,-}$ [see Fig. \ref{N-icon-2nd}]. Repeat the same procedure for the sets $E(\{m_1,m_5\},\{m_b\})$ and $E(\{m_2,m_5\},\{m_b\})$. The second term, denoted as $T_2$, is obtained by summing all these quantities with a multiplicative factor of $\frac{1}{2^2}$:
\begin{align}
    T_2=\frac{1}{2^2}&\left\lbrace\cos^2 \left(\frac{\Phi_{13}+\Phi_{23}}{2}t\right)\cos^2 \left(\frac{\Phi_{14}+\Phi_{24}}{2}t\right)\cos^2 \left(\frac{\Phi_{16}+\Phi_{26}}{2}t\right)\right.\nonumber\\
    &+\left.\cos^2 \left(\frac{\Phi_{13}-\Phi_{23}}{2}t\right)\cos^2 \left(\frac{\Phi_{14}-\Phi_{24}}{2}t\right)\cos^2 \left(\frac{\Phi_{16}-\Phi_{26}}{2}t\right)\right.\nonumber\\
    &\left.+\cos^2 \left(\frac{\Phi_{23}+\Phi_{35}}{2}t\right)\cos^2 \left(\frac{\Phi_{24}+\Phi_{45}}{2}t\right)\cos^2 \left(\frac{\Phi_{26}+\Phi_{56}}{2}t\right)\right.\nonumber\\
    &\left.+\cos^2 \left(\frac{\Phi_{23}-\Phi_{35}}{2}t\right)\cos^2 \left(\frac{\Phi_{24}-\Phi_{45}}{2}t\right)\cos^2 \left(\frac{\Phi_{26}-\Phi_{56}}{2}t\right)\right.\nonumber\\
    &+\left.\cos^2 \left(\frac{\Phi_{35}+\Phi_{13}}{2}t\right)\cos^2 \left(\frac{\Phi_{45}+\Phi_{14}}{2}t\right)\cos^2 \left(\frac{\Phi_{56}+\Phi_{16}}{2}t\right)\right.\nonumber\\
    &+\left.\cos^2 \left(\frac{\Phi_{35}-\Phi_{13}}{2}t\right)\cos^2 \left(\frac{\Phi_{45}-\Phi_{14}}{2}t\right)\cos^2 \left(\frac{\Phi_{56}-\Phi_{16}}{2}t\right)\right\rbrace.
\end{align}

\begin{figure}[t]
    \centering
    \includegraphics[width=\textwidth]{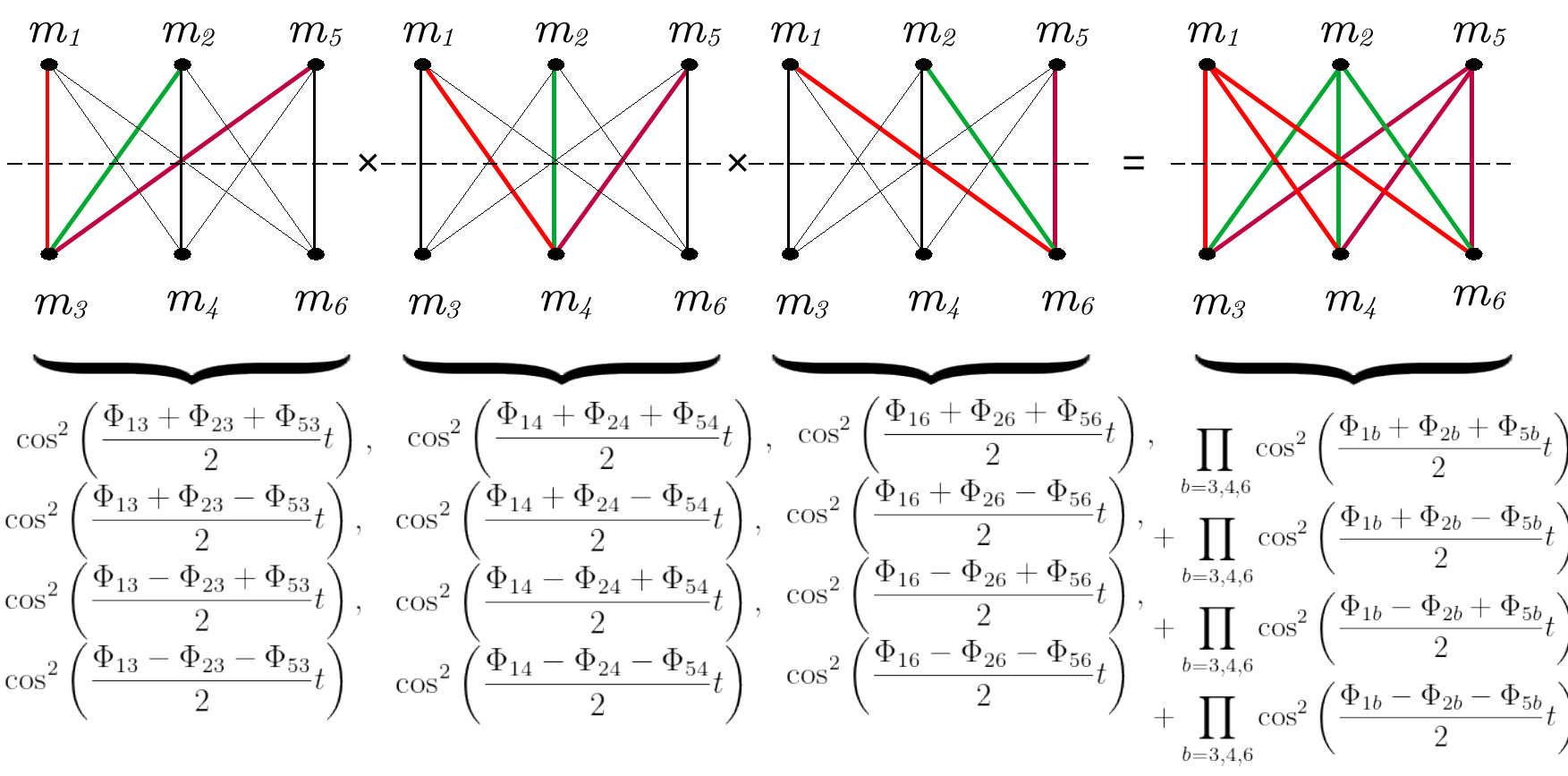}
    \caption{Diagram for calculating the third term in the series for calculating the I-concurrence for the $(1,2,5)|(3,4,6)$ bipartition}
    \label{N-icon-3rd}
\end{figure}

\noindent\textit{Third term:} Now, consider the set $E(V_3,\{m_b\})$ where $m_b\in V'_3$. With this set associate four functions: $\Tilde{f}_{b,++}=\cos^2\left(\frac{\Phi_{1b}+\Phi_{2b}+\Phi_{5b}}{2}t\right)$, $\Tilde{f}_{b,+-}=\cos^2\left(\frac{\Phi_{1b}+\Phi_{2b}-\Phi_{5b}}{2}t\right)$, $\Tilde{f}_{b,-+}=\cos^2\left(\frac{\Phi_{1b}-\Phi_{2b}+\Phi_{5b}}{2}t\right)$, and $\Tilde{f}_{b,--}=\cos^2\left(\frac{\Phi_{1b}-\Phi_{2b}-\Phi_{5b}}{2}t\right)$. Take the product of $\Tilde{f}_{b,++}$ for different values of $b=3,4,6$, and do the same for each of $\Tilde{f}_{b,+-}$, $\Tilde{f}_{b,-+}$ and $\Tilde{f}_{b,--}$ separately. Thereafter, add the four products: $\Tilde{f}_{3,++}\Tilde{f}_{4,++}\Tilde{f}_{6,++}+\Tilde{f}_{3,+-}\Tilde{f}_{4,+-}\Tilde{f}_{6,+-}+\Tilde{f}_{3,-+}\Tilde{f}_{4,-+}\Tilde{f}_{6,-+}+\Tilde{f}_{3,--}\Tilde{f}_{4,--}\Tilde{f}_{6,--}$ [see Fig. \ref{N-icon-3rd}]. The third term is simply given by the resultant quantity with a multiplicative factor of $\frac{1}{2^3}$:
\begin{align}
    T_3=\frac{1}{2^3}&\left\lbrace\cos^2 \left(\frac{\Phi_{13}+\Phi_{23}+\Phi_{35}}{2}t\right)\cos^2 \left(\frac{\Phi_{14}+\Phi_{24}+\Phi_{45}}{2}t\right)\cos^2 \left(\frac{\Phi_{16}+\Phi_{26}+\Phi_{56}}{2}t\right)\right.\nonumber\\
    &+\left.\cos^2 \left(\frac{\Phi_{13}+\Phi_{23}-\Phi_{35}}{2}t\right)\cos^2 \left(\frac{\Phi_{14}+\Phi_{24}-\Phi_{45}}{2}t\right)\cos^2 \left(\frac{\Phi_{16}+\Phi_{26}-\Phi_{56}}{2}t\right)\right.\nonumber\\
    &+\left.\cos^2 \left(\frac{\Phi_{13}-\Phi_{23}+\Phi_{35}}{2}t\right)\cos^2 \left(\frac{\Phi_{14}-\Phi_{24}+\Phi_{45}}{2}t\right)\cos^2 \left(\frac{\Phi_{16}-\Phi_{26}+\Phi_{56}}{2}t\right)\right.\nonumber\\
    &+\left.\cos^2 \left(\frac{\Phi_{13}-\Phi_{23}-\Phi_{35}}{2}t\right)\cos^2 \left(\frac{\Phi_{14}-\Phi_{24}-\Phi_{45}}{2}t\right)\cos^2 \left(\frac{\Phi_{16}-\Phi_{26}-\Phi_{56}}{2}t\right)\right\rbrace.
\end{align}

\noindent The I-concurrence for the considered bipartition is then given by
\begin{align}
    \mathcal{C}_{125|346}(t)=\sqrt{\frac{7}{4}-\frac{1}{2}(T_1+T_2+T_3)}.
\end{align}
We can calculate the I-concurrence for any bipartition of an $N$-body system pictorially using the graph representation in a similar manner.

Now, let us review some general properties of graphs that will be useful to study the entanglement properties of $|\Psi(t)\rangle_{12\ldots N}$. In a graph, two nodes are said to be connected if there exists a path, a sequence of edges, between them. This means that even if two nodes lack a direct edge between them, they can still be considered connected as long as there are intermediate nodes and edges forming a path. The number of edges within a path is called the length of the path. If there is a connection (path) between any two nodes within a given graph, that graph is referred to as a \textit{connected graph} \cite{Diestel2017}. Moreover, if a graph with more than $k$ nodes remains connected after removing any set of fewer than $k$ nodes, it is called \textit{$k$-connected}. It follows that, if a graph is $k$-connected, it is also $k'$-connected for any $k'<k$. The \textit{connectivity} of a graph, denoted as $\kappa(G)$, represents the largest value of $k$ for which the graph is $k$-connected. Connected graphs have a minimum connectivity of 1, i.e., $\kappa(G)\geq 1$, while disconnected graphs have $\kappa(G)=0$.

We now discuss the entanglement properties of the time-evolved $N$-body state (\ref{N-body-time-evolved-state})

\begin{proposition}[Genuine $N$-body entanglement]
    The time-evolution leads to genuine $N$-body entanglement if and only if the graph representing the system of $N$ masses is connected.
\end{proposition}

\begin{proof}
    From Eqs. (\ref{N-icon}) and (\ref{series}), it follows that I-concurrence across a given bipartition is nonzero, iff the entangling phase of at least one mass-pair with masses in the alternative parts of the bipartition is nonzero. In terms of the graphical representation, this implies that $E(V_k,V-V_k)\neq \varnothing$ for the corresponding $V_k|V-V_k$ bipartition of the graph. For genuine $N$-body entanglement, we require nonzero I-concurrence for every bipartition of the system, i.e., $E(U,V-U)\neq \varnothing~\forall~U\subset V$.

    Now suppose the graph $G(V,E)$ representing the system of $N$ masses is disconnected. This means there exist at least two nodes, say $m_p$ and $m_q$, that are not connected. We define the set of all nodes connected to $m_p$ as $S_{m_p}\subset V$, and the set of all nodes connected to $m_q$ as $S_{m_q}\subset V$. Since $m_p$ and $m_q$ are disconnected, we have $S_{m_p}\cap S_{m_q}=\varnothing$ and $E(S_{m_p},S_{m_q})=\varnothing$. If $S_{m_p}\cup S_{m_q}=V$, then we obtain a bipartition $S_{m_p}|S_{m_q}$ for which the I-concurrence is zero at all times $t>0$, leading to the absence of genuine $N$-body entanglement. On the other hand, if $S_{m_p}\cup S_{m_q}\neq V$, then $S'=V-(S_{m_p}\cup S_{m_q})$ contains nodes that are neither connected to $m_p$ nor to $m_q$. Therefore, $S'$ is disjoint from both $S_{m_p}$ and $S_{m_q}$: $S'\cap S_{m_p}=S'\cap S_{m_q}=\varnothing$, and $E(S_{m_p},S')=E(S_{m_q},S')=\varnothing$. Since both $E(S_{m_p},S_{m_q})=\varnothing$ and $E(S_{m_p},S')=\varnothing$, we have $E(S_{m_p},S_{m_q}\cup S')=\varnothing$. Consequently, we get a bipartition $S_{m_p}|(S_{m_q}\cup S')$ for which the I-concurrence vanishes for all times $t>0$, and genuine $N$-body entanglement does not emerge. This shows that for genuine entanglement, the corresponding graph must \textit{necessarily} be connected.

    On the other hand, in a connected graph, there does not exist any bipartition $U|V-U$ such that $E(U,V-U)=\varnothing$. If this were the case, the nodes in $U$ would be disconnected from the nodes in $V-U$, and the graph would be disconnected as well. This implies that in a connected graph $\nexists$ a bipartition for which the I-concurrence is zero at all times $t>0$. This shows that a connected graph is \textit{sufficient} for the creation of genuine entanglement.
\end{proof}

Note that the minimal requirement for a connected graph is that it must be at least 1-connected ($\kappa(G)\geq 1$). Additionally, the minimum number of edges in a 1-connected graph with $N$ nodes is $(N-1)$. Together, this implies that the minimal condition for the generation of genuine $N$-body entanglement is that at least $(N-1)$ entangling phases must be nonzero and the corresponding graph is connected. Importantly, the mere requirement of $(N-1)$ nonzero entangling phases is not sufficient, as the corresponding graph may still not be connected.\footnote{In case of 3-body system, the minimal requirement of 2 nonzero entangling phases (Proposition \ref{prop:3GME}) is sufficient for the corresponding graph to be connected.}

\begin{proposition}[$N$-qubit $GHZ$ entanglement]
    If all ${N\choose 2}$ entangling phases of an $N$-body system are nonzero and the ratio of any two of them is a ratio of two odd integers, then the state (\ref{N-body-time-evolved-state}) periodically becomes a generalised $N$-qubit GHZ state with 1 ebit of entanglement across every bipartition. %The time period of this periodic behaviour is $T=\frac{2\pi}{\Phi}$.
\end{proposition}

\begin{proof}
    Let $\Phi_{xy}=(2n_{xy}+1)\Phi$, $n_{xy}=0,1,2,\ldots$, $\forall x<y\in\{1,2,\ldots,N\}$ and $\Phi\neq0$. This naturally implies that the ratio of any two entangling phases forms a ratio of two odd integers: $\frac{\Phi_{xy}}{\Phi_{x'y'}}=\frac{2n_{xy}+1}{2n_{x'y'}+1}$. Eq. (\ref{series}) then takes the form 
    \begin{small}
    \begin{align}\label{seriesNGHZ}
    \Lambda(t)=&\frac{1}{2} \sum_{a=p_1}^{p_k} \left[\prod_{b=p_{k+1}}^{p_N} \cos^2\left(\frac{(2n_{ab}+1)}{2}\Phi t \right)\right]\nonumber\\
    &+\frac{1}{2^2}\sum_{s_1=0}^{1}\sum_{\substack{a_1<a_2,\\a_1,a_2={p_1}}}^{p_k} \left[\prod_{b=p_{k+1}}^{p_N} \cos^2\left(\frac{(2n_{a_1b}+1)+(-1)^{s_1}(2n_{a_2b}+1)}{2}\Phi t \right)\right]\nonumber\\
    &+\frac{1}{2^3}\sum_{s_1,s_2=0}^1\sum_{\substack{a_1<a_2<a_3\\a_1,a_2,a_3={p_1}}}^{p_k} \left[\prod_{b=p_{k+1}}^{p_N} \cos^2\left(\frac{(2n_{a_1b}+1)+(-1)^{s_1}(2n_{a_2b}+1)+(-1)^{s_2}(2n_{a_3b}+1)}{2}\Phi t\right)\right]+\ldots\nonumber\\
    &+\frac{1}{2^k}\sum_{s_1,\ldots ,s_{(k-1)}=0}^1\sum_{\substack{a_1<a_2<\ldots<a_k\\a_1,a_2,\ldots,a_k=p_1}}^{p_k} \left[\prod_{b=p_{k+1}}^{p_N} \cos^2\left(\frac{(2n_{a_1b}+1)+\ldots+(-1)^{s_{(k-1)}}(2n_{a_kb}+1)}{2}\Phi t \right)\right].
\end{align}
\end{small}

Each of the $k$ terms in the series involves the summation of products of several $\cos^2(\cdot)$ functions. Observe that in the odd-numbered terms (1st, 3rd, 5th, and so on), the argument of each $\cos^2(\cdot)$ function is an odd multiple of $\frac{\Phi t}{2}$. Whereas, in the even-numbered terms (2nd, 4th, 6th, and so on), the argument of each $\cos^2(\cdot)$ function is an even multiple of $\frac{\Phi t}{2}$. These follow from the mathematical properties that the algebraic sum (summation with signs) of an odd number of odd integers always results in an odd integer, while that of an even number of odd integers yields an even integer.\footnote{This can be proved straightforwardly by noting that the sum/difference of two odd integers is an even integer, and the sum/difference of one odd and one even integer is an odd integer.}

Now, for $\Phi t=(2n+1)\pi, n=0,1,2,\ldots$, the arguments of $\cos^2(\cdot)$ functions in the odd-numbered terms become odd multiples of $\frac{\pi}{2}$ and the functions vanish. Arguments of $\cos^2(\cdot)$ functions in the even-numbered terms, on the other hand, become even multiples of $\frac{\pi}{2}$, making the functions 1.

So, at $t=\frac{(2n+1)\pi}{\Phi}, n=0,1,2,\ldots$, all the odd-numbered terms become zero and the even-numbered terms are given by
\begin{align}
    T_l = \frac{1}{2^l}2^{l-1}{k \choose l}=\frac{1}{2}{k\choose l},~~\text{for}~~l=2,4,6,\ldots. \nonumber
\end{align}
Therefore, (\ref{seriesNGHZ}) can be expressed as
\begin{align}
    \Lambda\left[t=\frac{(2n+1)\pi}{\Phi}\right]&=\frac{1}{2}\left\{{k\choose 2}+{k\choose 4}+\ldots+{k\choose {k-1}}\right\},~~\text{if $k$ is odd,}\\
    \text{and,}~~~\Lambda\left[t=\frac{(2n+1)\pi}{\Phi}\right]&=\frac{1}{2}\left\{{k\choose 2}+{k\choose 4}+\ldots+{k\choose {k-2}}+{k\choose k}\right\},~~\text{if $k$ is even.}
\end{align}
However, for both cases, we can write
\begin{align}\label{seriesclosed}
    \Lambda\left[t=\frac{(2n+1)\pi}{\Phi}\right]&=\frac{1}{2}\left\{{{k-1}\choose 1}+{{k-1}\choose 2}+\ldots+{{k-1}\choose {k-2}}+{{k-1}\choose {k-1}}\right\}\nonumber\\
    &=\frac{2^{k-1}-1}{2},
\end{align}
where we have used the following identities
\begin{align*}
    {{k}\choose l}&={{k-1}\choose {l-1}}+{{k-1}\choose l},\\
    {k\choose k}&={{k-1}\choose {k-1}}=1.
\end{align*}

Putting (\ref{seriesclosed}) in (\ref{N-icon}), we get the I-concurrence for any $k|(N-k)$ bipartition at $t=\frac{(2n+1)\pi}{\Phi},n=0,1,2,\ldots$, which is
\begin{align}
    \mathcal{C}_{p_1p_2\ldots p_k|p_{k+1}\ldots p_N}\left(t=\frac{(2n+1)\pi}{\Phi}\right)=1,~~\forall~1\leq k\leq \left\lfloor\frac{N}{2}\right\rfloor.
\end{align}

This implies that the state (\ref{N-body-time-evolved-state}) periodically evolves into a generalised $N$-qubit GHZ state with 1 ebit of entanglement across every bipartition. Time period of this periodic behaviour is $T=\frac{2\pi}{\Phi}$.
\end{proof}

%-------------------------------------------
\subsection*{Multipartite entanglement}
%-------------------------------------------
Full quantification of multipartite entanglement of a generic $N$-body system is hard. This is because with the increasing number of subsystems in a many-body system, the number of independent entanglement measures, each focusing on a different aspect of multipartite entanglement, increases exponentially.

However, since we are interested in multipartite entanglement, a family of multipartite entanglement measures \cite{PhysRevA.69.052330} that are generalisations of the Meyer-Wallach measure \cite{meyer2002global, Brennen2003}, turns out to be particularly useful and relevant for our purpose. These measures are easily computable and applicable to pure states involving any number of qudits ($d$-dimensional quantum systems). Furthermore, they exhibit a close relationship with the previously discussed bipartite entanglement measures.

Specifically, for an $N$-qudit pure state $|\psi\rangle\in(\mathbb{C}^d)^{\otimes N}$, the family (parameterised by $k$) of multipartite entanglement measures can be defined as \cite{PhysRevA.69.052330}:
\begin{align}\label{Qk}
    Q_k=\frac{d^k}{d^k-1}\left(1-\frac{k!(N-k)!}{N!}\sum_{\vert S\vert=k}\Tr\rho_{S}^2\right),~~\text{for}~1\leq k\leq\left\lfloor\frac{N}{2}\right\rfloor,
\end{align}
where $\rho_S=\Tr_{S'}\oper{\psi}$ represents the reduced density matrix of a subsystem consisting of $k$ qudits obtained by tracing out $(N-k)$ qudits, and the sum is taken over all possible reduced $k$-qudit-subsystems of $|\psi\rangle$. For each $k$, $Q_k$ serves as an entanglement monotone, with values in the range $0\leq Q_k\leq 1$, and $Q_k=0$ if and only if $|\psi\rangle$ is a fully separable state. Notably, when $d=2$ and $k=1$, $Q_k$ corresponds to the well-known Meyer-Wallach measure \cite{meyer2002global, Brennen2003}.

Essentially, $Q_k$ measures the average bipartite entanglement, over all possible $k|(N-k)$ bipartitions of $|\psi\rangle$ for a fixed $k$. In terms of I-concurrence (see Eq. \ref{Icon_k|N-k} in Appendix \ref{app} for the mathematical definition of I-concurrence for $k|(N-k)$ bipartition), (\ref{Qk}) can be expressed as
\begin{align}
    Q_k=\frac{d^k}{2(d^k-1)}\frac{1}{{N\choose k}}\sum_{\{p_1,\ldots, p_k\}} \mathcal{C}^2_{p_1\ldots p_k|p_{k+1}\ldots p_N},
\end{align}
where $p_1<p_2<\ldots<p_k\in\{1,2,\ldots,N\}$, $p_{k+1}<\ldots<p_N\in\{1,2,\ldots,N\}\backslash\{p_1,\ldots, p_k\}$, and the sum is over all possible sets of $\{p_1,\ldots,p_k\}$.

So, we can compute $Q_k$ of the state (\ref{N-body-time-evolved-state}) for any $1\leq k\leq\left\lfloor\frac{N}{2}\right\rfloor$ as a function of time and entangling phases by simply setting $d=2$ and substituting (\ref{N-icon}) and (\ref{series}) in the above expression. In particular, the Meyer-Wallach measure of (\ref{N-body-time-evolved-state}) is given by
\begin{align}\label{Q}
    Q_1(t) &=\frac{1}{N} \sum_{p_1=1}^N \mathcal{C}^2_{p_1|p_2\ldots p_N}(t)\nonumber\\
    &= 1-\frac{1}{N}\sum_{p_1=1}^N\left[\prod_{b=p_2}^{p_N}\cos^2\left(\frac{\Phi_{p_1b}t}{2}\right)\right].
\end{align}

\section{\label{s5} Comparing one-vs-rest entanglement for different $N$}

We now turn our attention to the $1|(N-1)$ (one-vs-rest) bipartitions in an $N$-body system. We have already derived the I-concurrence expression for these bipartitions [See (\ref{one-vs-rest})], which measures the amount of bipartite entanglement a single mass shares with the rest $(N-1)$ masses. For $N=2$, this quantity corresponds to the concurrence (\ref{2body_concurrence}) measuring entanglement between two qubits $m_1$ and $m_2$. In the case of $N=3$, it measures the entanglement between a qubit $m_p$ and two qubits $(m_q,m_r)$, for $p\neq q\neq r\in\{1,2,3\}$. Similarly, for $N=4$, it measures the entanglement between one qubit and three qubits, and so forth. 

We wish to compare these I-concurrence values for a fixed bipartition across different values of $N$. Specifically, we aim to understand how the entanglement properties of a particular one-vs-rest bipartition change when we vary the number of masses in the `rest' component, while keeping all other parameters, such as the geometry and superposition orientations of the other masses, unchanged.

Consider a particular $1|(N-1)$ bipartition $m_{p_1}|m_{p_2}\ldots m_{p_N}$ in an $N$-body system ($N\geq 3$). The I-concurrence for this bipartition, at any instant $t>0$, is lower bounded as
\begin{align}
    \mathcal{C}_{p_1|p_2\ldots p_N}(t)\geq \max_{p_k}\left\{\mathcal{C}_{p_1|p_2\ldots p_{k-1}p_{k+1}\ldots p_N}(t)\right\},
\end{align}
where $\mathcal{C}_{p_1|p_2\ldots p_{k-1}p_{k+1}\ldots p_N}(t)$ is the I-concurrence for the $1|(N-2)$ bipartition\\ $m_{p_1}|m_{p_2}\ldots m_{p_{k-1}}m_{p_{k+1}}\ldots m_{p_N}$ of the corresponding $(N-1)$-body system, i.e., the bipartition without the mass $m_{p_k}$ for $2\leq k\leq N$ in the  `rest' part. Importantly, note here that $\mathcal{C}_{p_1|p_2\ldots p_{k-1}p_{k+1}\ldots p_N}(t)$ is the I-concurrence of the time evolved pure state of the $(N-1)$-body system without $m_{p_k}$, for the said bipartition, \textit{not} of the reduced state of the $(N-1)$-body system obtained by tracing $m_{p_k}$.

This relation follows directly from the mathematical expression (\ref{one-vs-rest}) of the I-concurrence:
\begin{align}
     &\prod_{b=p_2}^{p_N} \cos^2 \left(\frac{\Phi_{p_1b}t}{2}\right)\leq\prod_{b=p_2}^{p_{k-1}} \cos^2 \left(\frac{\Phi_{p_1b}t}{2}\right)\prod_{b=p_{k+1}}^{p_{N}} \cos^2 \left(\frac{\Phi_{p_1b}t}{2}\right)~~~\left[\because 0\leq \cos^2 \left(\frac{\Phi_{p_1b}t}{2}\right)\leq1,\forall b\right],\nonumber\\
    &\Rightarrow 1-\prod_{b=p_2}^{p_N} \cos^2 \left(\frac{\Phi_{p_1b}t}{2}\right)\geq 1-\prod_{b=p_2}^{p_{k-1}} \cos^2 \left(\frac{\Phi_{p_1b}t}{2}\right)\prod_{b=p_{k+1}}^{p_{N}} \cos^2 \left(\frac{\Phi_{p_1b}t}{2}\right),\nonumber\\
    &\text{i.e.,}~\mathcal{C}_{p_1|p_2\ldots p_N}(t)\geq\mathcal{C}_{p_1|p_2\ldots p_{k-1}p_{k+1}\ldots p_{N}}(t).\nonumber
\end{align}

Now, note that $\mathcal{C}_{p_1|p_2\ldots p_{k-1}p_{k+1}\ldots p_{N}}(t)=0$ if 
\begin{align}
    \prod_{b=p_2}^{p_{k-1}} \cos^2 \left(\frac{\Phi_{p_1b}t}{2}\right)\prod_{b=p_{k+1}}^{p_{N}} \cos^2 \left(\frac{\Phi_{p_1b}t}{2}\right)=1,\nonumber
\end{align}
that is, when $t=\frac{2n_{p_1b}\pi}{\Phi_{p_1b}}$, for $n_{p_1b}=0,1,2,\ldots$ and $b\in\{p_2,\ldots,p_N\}\backslash\{p_k\}$. However, if at that instant $t$, $\frac{\Phi_{p_1p_k}t}{2}\neq2n_{p_1p_k}\pi$, for some $n_{p_1p_k}=0,1,2,\ldots$, then $\mathcal{C}_{p_1|p_2\ldots p_{N}}(t)\neq0$. This implies that, in general, we observe a longer duration of nonzero entanglement across one-vs-rest bipartitions in $N$-body systems compared to $(N-1)$-body systems.

%----------------------------------------------------------------------
\section{\label{s6} Conclusions}
%----------------------------------------------------------------------
Recent studies have demonstrated that the quantum nature of the gravitational field arising from two (distant) test masses, each existing in a spatially superposed quantum state, can establish entanglement between them.  This phenomenon of entanglement generation has been rigorously investigated within the framework of linearised quantum gravity.  In this paper, we studied this phenomenon of quantum gravity induced entanglement of masses in many-body systems. We considered $N\geq 3$ test masses, each initially prepared in a superposition of two non-overlapping spatially localised states, interacting through their mutual gravity. So in our approach, we have treated each mass as an  ``effective"  qubit in the position basis. For describing the pairwise interaction Hamiltonian we have used Newtonian potential, a function of position observables. To allow for complete generality in our treatment, we maintained an entirely arbitrary setup. In other words, we neither prescribed a specific geometry for the arrangement of the masses nor assumed any particular orientation for the superpositions of the masses.  This means that in our scenario, the definition of two orthogonal spatial states of $i^{\text{th}}$ test mass is
independent of, and generally differs from, the definition of orthogonal spatial states of $j^{\text{th}}$ test mass.  Our primary objective was to explore the characteristics of many-body entanglement arising from the mutual gravitational interactions among the masses.

We studied the entanglement properties of the time-evolved state by computing the entanglement across all possible bipartitions. For each bipartition, we derived an exact expression of I-concurrence, a measure of bipartite entanglement, in terms of the relevant entangling phases. Additionally, we quantified the degree of multipartite entanglement by evaluating a set of generalised Meyer-Wallach measures as functions of the entangling phases.

In the case of a two-body system, the entangling phase serves as the fundamental quantity encapsulating information about the system's configuration and entanglement. Likewise, for $N$-body systems, the collection of entangling phases, each associated with a mass-pair within the system, collectively captures details about the geometry and orientations of the masses. By expressing all entanglement-related quantities in terms of these entangling phases, we achieved a comprehensive treatment that enabled us to draw conclusions for any given system setup.

We examined the case of three masses in detail. The nature of three-body entanglement is characterised by establishing the conditions related to the entangling phases. For instance, we showed that the time-evolved state is genuine three-body entangled if and only if at least two of the three entangling phases are nonzero, and further, if their ratio is that of two odd integers, the state periodically becomes a GHZ-type state (maximally entangled across each bipartition). Moreover, the system remains genuinely entangled at all times $t > 0$, never becoming biseparable if the ratio of no two entangling phases is a rational number. These findings highlight that the characteristics of many-body entanglement significantly depend on the system setup in a non-trivial manner, in contrast to the case of two masses, where entanglement generically oscillates between 0 and 1 ebit, irrespective of the setup.

In the context of generic $N$-body systems, we introduced a diagrammatic approach for calculating the I-concurrence for any given bipartition. To do this, we used graphical representation, with the masses acting as nodes and the nonzero entangling phases forming the edges. We illustrated this method with a simple example. Furthermore, this graphical representation is also valuable in discussing the many-body entanglement characteristics---the $N$-body system can become genuinely entangled provided the corresponding graph is connected.

We demonstrated that if all the entangling phases are nonzero, and the ratio of any two forms the ratio of two odd integers, the time-evolved state periodically becomes a generalised GHZ-type state. However, while the condition of all nonzero entangling phases is sufficient, it may not be strictly necessary. This becomes evident in the three-body scenario where only two of three entangling phases were needed to be nonzero (the third can be zero) for the periodic GHZ state to manifest. The question of whether a generalised GHZ state can be achieved periodically with the minimal condition for genuine many-body entanglement---i.e., requiring $(N-1)$ nonzero entangling phases that correspond to a connected graph---along with the condition that the ratio of any two nonzero entangling phases forms the ratio of two odd integers, remains open.

Another interesting observation is that at any given instant of time, the amount of entanglement across a $1|(N-1)$ bipartition in an $N$-body system is greater than or equal to the amount of entanglement across the $1|(N-2)$ bipartition in the corresponding $(N-1)$-body system obtained by removing a mass. This implies that, in general, we can increase the period of nonzero entanglement hierarchically by adding more masses. For example, when dealing with just two masses, $m_1$ and $m_2$, their entanglement oscillates with a particular time period. However, by introducing a third mass, $m_3$, the entanglement between $m_1$ (or $m_2$) and the other two masses remains nonzero for a longer period of time. Similarly, if we introduce a fourth mass, $m_4$, the entanglement between $m_1$ and $(m_2,m_3,m_4)$ sustains even longer, and so forth. This observation holds practical significance for experiments. The current Stern-Gerlach-based experiment \cite{Bose2017PRL} is hindered by challenges in maintaining masses in spatial superposition for extended durations, limiting the time window for entanglement generation and detection. Future experimental schemes that wish to eliminate this constraint may benefit from the extended periods of nonzero entanglement between masses, facilitating the observation and study of quantum gravity-induced entanglement.

\appendix
\section{I-concurrence of time-evolved $N$-body state for arbitrary $k|(N-k)$ bipartition}
\label{app}
Consider the $k|(N-k)$ bipartition $(m_1,m_2,\ldots,m_k)|(m_{k+1},m_{k+2},\ldots,m_N)$ of the $N$-body system, for any $1\leq k\leq \left\lfloor\frac{N}{2}\right\rfloor$. The state (\ref{N-body-time-evolved-state}) can be written in this bipartition as
\begin{align}
    &|\Psi(t)\rangle_{12\ldots N}\nonumber\\
    &=\frac{1}{2^{\frac{N}{2}}}\left[\sum_{j_{k+1},\ldots,j_N}\left(\sum_{j_{1},\ldots,j_{k}}e^{i\phi_{j_1j_2\ldots j_N}t}|j_1\rangle|j_2\rangle\ldots|j_k\rangle \right) |j_{k+1}\rangle|j_{k+2}\rangle\ldots|j_N\rangle\right],\nonumber\\
    &=\frac{1}{2^{\frac{N}{2}}}\left[\sum_{j_{k+1},\ldots,j_{N}} |\Theta_{j_{k+1}\ldots j_N}(t)\rangle|j_{k+1}\rangle|j_{k+2}\rangle\ldots|j_N\rangle \right],
\end{align}
where, $|\Theta_{j_{k+1}\ldots j_N}(t)\rangle= \sum_{j_{1},\ldots,j_{k}}e^{i\phi_{j_1j_2\ldots j_N}t}|j_1\rangle|j_2\rangle\ldots|j_k\rangle$.

I-concurrence of (\ref{N-body-time-evolved-state}) for this bipartition is given by \cite{Rungta2001}
\begin{align}\label{Icon_k|N-k}
    \mathcal{C}_{12\ldots k|(k+1)\ldots N}(t)=\sqrt{2\left(1-\Tr\rho_{12\ldots k}^2(t)\right)},
\end{align}
where, 
\begin{align}
    \rho_{12\ldots k}(t)&=\Tr_{(k+1),\ldots,N}\left[|\Psi(t)\rangle_{12\ldots N}\langle \Psi(t)|\right] \nonumber\\
    &=\frac{1}{2^N} \left(\sum_{j_{k+1},\ldots,j_N} |\Theta_{j_{k+1}\ldots j_N}(t)\rangle\langle\Theta_{j_{k+1}\ldots j_N}(t)| \right)
\end{align}
is the reduced density matrix of the subsystem containing $k$ masses. 

Now,
\begin{align}
    \Tr\rho^2_{12\ldots k}(t)&=\sum_{j_1,\ldots,j_k}\sum_{j'_1,\ldots,j'_k}\left\vert\langle j_1j_2\ldots j_k|\rho_{12\ldots k}(t)|j'_1j'_2\ldots j'_k\rangle\right\vert^2,\nonumber\\
    &=\frac{1}{2^{2N}}\sum_{j_1,\ldots,j_k}\sum_{j'_1,\ldots,j'_k}\left\vert \sum_{j_{k+1},\ldots,j_N}e^{i\left(\phi_{j_1\ldots j_kj_{k+1}\ldots j_N}-\phi_{j'_1\ldots j'_kj_{k+1}\ldots j_N}\right)t}\right\vert^2,
\end{align}
where $j_a,j'_a\in\{0_a,1_a\}$ for $a=1,2,\ldots,k$, and $|j_1\ldots j_k\rangle=|j_1\rangle|j_2\rangle\ldots|j_k\rangle$. We can write the above expression in a slightly different form as
\begin{align}
    \Tr\rho^2_{12\ldots k}(t) &=\frac{1}{2^{2N}}\sum_{x=0}^{2^k-1}\sum_{y=0}^{2^k-1}\left\vert \sum_{j_{k+1},\ldots,j_N}e^{i\left(\phi_{\mathcal{B}(x)j_{k+1}\ldots j_N}-\phi_{\mathcal{B}(y)j_{k+1}\ldots j_N}\right)t}\right\vert^2,\nonumber\\
    &=\frac{1}{2^{2N}}\sum_{x=0}^{2^k-1}\sum_{y=0}^{2^k-1}\abs{\lambda_{xy}}^2
\end{align}
where $\mathcal{B}(x)=j_1j_2\ldots j_k$ and $\mathcal{B}(y)=j'_1j'_2\ldots j'_k$ are the binary representations of $x,y\in[0,2^k-1]$, respectively, as $k$-bit strings,\footnote{For example, the binary representation of 3 as a 5-bit string is $\mathcal{B}(3)=00011$.} and $\lambda_{xy}=\sum_{j_{k+1}\ldots j_N}e^{i\left(\phi_{\mathcal{B}(x)j_{k+1}\ldots j_N}-\phi_{\mathcal{B}(y)j_{k+1}\ldots j_N}\right)t}$.

Now, note that $\lambda_{xy}^*=\lambda_{yx}$, implying $\abs{\lambda_{xy}}^2=\abs{\lambda_{yx}}^2$. Additionally, $\lambda_{xy}=2^{(N-k)}$, for $x=y$. We can then split $\sum_{x,y}\abs{\lambda_{xy}}^2$ as
\begin{align}
    \sum_{x,y}\abs{\lambda_{xy}}^2&=\sum_{x=y}\abs{\lambda_{xy}}^2+2\sum_{x<y}\abs{\lambda_{xy}}^2\nonumber\\
    &=\sum_{x=0}^{2^k-1}2^{2(N-k)}+2\sum_{x<y}\abs{\lambda_{xy}}^2.\nonumber
\end{align}
Therefore, we get
\begin{align}
    \Tr\rho^2_{12\ldots k}(t) = \frac{1}{2^{2N}}\left(2^{2(N-k)}\times2^k+2\sum_{y=0}^{2^k-1}\sum_{x=0}^{y-1}\left\vert \lambda_{xy}\right\vert^2\right).
\end{align}

Putting the above expression for $\Tr \rho_{12\ldots k}^2(t)$ in (\ref{Icon_k|N-k}), we get
\begin{align}\label{pIcon}
    \mathcal{C}_{12\ldots k|(k+1)\ldots N}(t)= \sqrt{\frac{2^k-1}{2^{k-1}}-\frac{4}{2^{2N}}\left(\sum_{y=0}^{2^k-1}\sum_{x=0}^{y-1}\vert\lambda_{xy}\vert^2\right)}.
\end{align}

Before evaluating $\sum_{y=0}^{2^k-1}\sum_{x=0}^{y-1}\abs{\lambda_{xy}}^2$, let us consider a simple example with $N=6$ and $k=3$ to get a typical idea of $\abs{\lambda_{xy}}^2$. For this case, $\lambda_{01}$ (say) is given as
\begin{align}
    \lambda_{01}&=\sum_{j_4,j_5,j_6} e^{i(\phi_{\mathcal{B}(0)j_4j_5j_6}-\phi_{\mathcal{B}(1)j_4j_5j_6})t} =\sum_{j_4,j_5,j_6} e^{i(\phi_{0_10_20_3j_4j_5j_6}-\phi_{0_10_21_3j_4j_5j_6})t} \nonumber\\
    &=e^{i(\phi_{0_10_20_30_40_50_6}-\phi_{0_10_21_30_40_50_6})t}\left(\sum_{j_4,j_5,j_6} e^{i(\phi_{0_10_20_3j_4j_5j_6}-\phi_{0_10_21_3j_4j_5j_6}-\phi_{0_10_20_30_40_50_6}+\phi_{0_10_21_30_40_50_6})t}\right). \nonumber
\end{align}
Now, using the definition $\phi_{j_1j_2\ldots j_N}=\sum_{p<q=1}^N \phi_{j_pj_q}$ [See Eq. (\ref{bw})], we get
\begin{align}
    \phi_{0_10_20_3j_4j_5j_6}-\phi_{0_10_21_3j_4j_5j_6} =&~ \phi_{0_10_2}+\phi_{0_10_3}+\phi_{0_1j_4}+\phi_{0_1j_5}+\phi_{0_1j_6}+\phi_{0_20_3}+\phi_{0_2j_4}+\phi_{0_2j_5}+\phi_{0_2j_6}\nonumber\\
    &+\phi_{0_3j_4}+\phi_{0_3j_5}+\phi_{0_3j_6}+\phi_{j_4j_5}+\phi_{j_4j_6}+\phi_{j_5j_6}\nonumber\\    
    &-\phi_{0_10_2}-\phi_{0_11_3}-\phi_{0_1j_4}-\phi_{0_1j_5}-\phi_{0_1j_6}-\phi_{0_21_3}-\phi_{0_2j_4}-\phi_{0_2j_5}-\phi_{0_2j_6}\nonumber\\
    &-\phi_{1_3j_4}-\phi_{1_3j_5}-\phi_{1_3j_6}-\phi_{j_4j_5}-\phi_{j_4j_6}-\phi_{j_5j_6}\nonumber\\
    =&~\phi_{0_10_3}+\phi_{0_20_3}+\phi_{0_3j_4}+\phi_{0_3j_5}+\phi_{0_3j_6}\nonumber\\
    &-\phi_{0_11_3}-\phi_{0_21_3}-\phi_{1_3j_4}-\phi_{1_3j_5}-\phi_{1_3j_6},\nonumber\\
    \phi_{0_10_20_30_40_50_6}-\phi_{0_10_21_30_40_50_6} =&~ \phi_{0_10_3}+\phi_{0_20_3}+\phi_{0_30_4}+\phi_{0_30_5}+\phi_{0_30_6}\nonumber\\
    &-\phi_{0_11_3}-\phi_{0_21_3}-\phi_{1_30_4}-\phi_{1_30_5}-\phi_{1_30_6},\nonumber\\
    \text{and,}~~ \phi_{0_10_20_3j_4j_5j_6}-\phi_{0_10_21_3j_4j_5j_6}&-\phi_{0_10_20_30_40_50_6}+\phi_{0_10_21_30_40_50_6} \nonumber\\
    = &~\phi_{0_3j_4}+\phi_{0_3j_5}+\phi_{0_3j_6}-\phi_{1_3j_4}-\phi_{1_3j_5}-\phi_{1_3j_6}\nonumber\\
    &-\phi_{0_30_4}-\phi_{0_30_5}-\phi_{0_30_6}+\phi_{1_30_4}+\phi_{1_30_5}+\phi_{1_30_6}.
\end{align}
Now, for $j_4=0_4,j_5=0_5,j_6=0_6$
\begin{align}
    \phi_{0_10_20_30_40_50_6}-\phi_{0_10_21_30_40_50_6}-\phi_{0_10_20_30_40_50_6}+\phi_{0_10_21_30_40_50_6}=0, \nonumber
\end{align}
for $j_4=0_4,j_5=0_5,j_6=1_6$
\begin{align}
    \phi_{0_10_20_30_40_51_6}-\phi_{0_10_21_30_40_51_6}-\phi_{0_10_20_30_40_50_6}&+\phi_{0_10_21_30_40_50_6}\nonumber\\
    &=\phi_{0_31_6}+\phi_{1_30_6}-\phi_{0_30_6}-\phi_{1_31_6}\nonumber\\
    &=\Phi_{36}, \nonumber
\end{align}
for $j_4=0_4,j_5=1_5,j_6=0_6$
\begin{align}
    \phi_{0_10_20_30_41_50_6}-\phi_{0_10_21_30_41_50_6}-\phi_{0_10_20_30_40_50_6}&+\phi_{0_10_21_30_40_50_6}\nonumber\\
    &=\phi_{0_31_5}+\phi_{1_30_5}-\phi_{0_30_5}-\phi_{1_31_5}\nonumber\\
    &=\Phi_{35}, \nonumber
\end{align}
for $j_4=0_4,j_5=1_5,j_6=1_6$
\begin{align}
    \phi_{0_10_20_30_41_51_6}-\phi_{0_10_21_30_41_51_6}-\phi_{0_10_20_30_40_50_6}&+\phi_{0_10_21_30_40_50_6}\nonumber\\
    &=\phi_{0_31_5}+\phi_{1_30_5}-\phi_{0_30_5}-\phi_{1_31_5}\nonumber\\
    &+\phi_{0_31_6}+\phi_{1_30_6}-\phi_{0_30_6}-\phi_{1_31_6}\nonumber\\
    &=\Phi_{35}+\Phi_{36}, \nonumber
\end{align}
for $j_4=1_4,j_5=0_5,j_6=0_6$
\begin{align}
    \phi_{0_10_20_31_40_50_6}-\phi_{0_10_21_31_40_50_6}-\phi_{0_10_20_30_40_50_6}&+\phi_{0_10_21_30_40_50_6}\nonumber\\
    &=\phi_{0_31_4}+\phi_{1_30_4}-\phi_{0_30_4}-\phi_{1_31_4}\nonumber\\
    &=\Phi_{34}, \nonumber
\end{align}
for $j_4=1_4,j_5=0_5,j_6=1_6$
\begin{align}
    \phi_{0_10_20_31_40_51_6}-\phi_{0_10_21_31_40_51_6}-\phi_{0_10_20_30_40_50_6}&+\phi_{0_10_21_30_40_50_6}\nonumber\\
    &=\phi_{0_31_4}+\phi_{1_30_4}-\phi_{0_30_4}-\phi_{1_31_4}\nonumber\\
    &+\phi_{0_31_6}+\phi_{1_30_6}-\phi_{0_30_6}-\phi_{1_31_6}\nonumber\\
    &=\Phi_{34}+\Phi_{36}, \nonumber
\end{align}
for $j_4=1_4,j_5=1_5,j_6=0_6$
\begin{align}
    \phi_{0_10_20_31_41_50_6}-\phi_{0_10_21_31_41_50_6}-\phi_{0_10_20_30_40_50_6}&+\phi_{0_10_21_30_40_50_6}\nonumber\\
    &=\phi_{0_31_4}+\phi_{1_30_4}-\phi_{0_30_4}-\phi_{1_31_4}\nonumber\\
    &+\phi_{0_31_5}+\phi_{1_30_5}-\phi_{0_30_5}-\phi_{1_31_5}\nonumber\\
    &=\Phi_{34}+\Phi_{35}, \nonumber
\end{align}
and for $j_4=1_4,j_5=1_5,j_6=1_6$
\begin{align}
    \phi_{0_10_20_31_41_51_6}-\phi_{0_10_21_31_41_51_6}-\phi_{0_10_20_30_40_50_6}&+\phi_{0_10_21_30_40_50_6}\nonumber\\
    &=\phi_{0_31_4}+\phi_{1_30_4}-\phi_{0_30_4}-\phi_{1_31_4}\nonumber\\
    &+\phi_{0_31_5}+\phi_{1_30_5}-\phi_{0_30_5}-\phi_{1_31_5}\nonumber\\
    &+\phi_{0_31_6}+\phi_{1_30_6}-\phi_{0_30_6}-\phi_{1_31_6}\nonumber\\
    &=\Phi_{34}+\Phi_{35}+\Phi_{36}, \nonumber
\end{align}
So,
\begin{align}
    \lambda_{01}&=e^{i(\phi_{0_10_20_30_40_50_6}-\phi_{0_10_21_30_40_50_6})t}\left( 1+e^{i\Phi_{34}t}+e^{i\Phi_{35}t}+e^{i\Phi_{36}t}+e^{i(\Phi_{34}+\Phi_{35})t}+e^{i(\Phi_{34}+\Phi_{36})t}\right.\nonumber\\
    &+\left.e^{i(\Phi_{35}+\Phi_{36})t}+e^{i(\Phi_{34}+\Phi_{35}+\Phi_{36})t}\right)\nonumber\\
    &=e^{i(\phi_{0_10_20_30_40_50_6}-\phi_{0_10_21_30_40_50_6})t}\left( 1+e^{i\Phi_{34}t}\right)\left( 1+e^{i\Phi_{35}t}\right)\left( 1+e^{i\Phi_{36}t}\right),\nonumber\\
    |\lambda_{01}|^2&=\left|1+e^{i\Phi_{34}t}\right|^2\left\vert1+e^{i\Phi_{35}t}\right\vert^2\left\vert 1+e^{i\Phi_{36}t}\right\vert^2\nonumber\\
    &=4\cos^2\left(\frac{\Phi_{34}}{2}t\right)4\cos^2\left(\frac{\Phi_{35}}{2}t\right)4\cos^2\left(\frac{\Phi_{36}}{2}t\right).
\end{align}
Similarly, it can be shown that
\begin{align}
    &\lambda_{24}\nonumber\\
    &=\sum_{j_4,j_5,j_6} e^{i(\phi_{0_11_20_3j_4j_5j_6}-\phi_{1_10_20_3j_4j_5j_6})t} \nonumber\\
    &=e^{i(\phi_{0_11_20_30_40_50_6}-\phi_{1_10_20_30_40_50_6})t}\left( 1+e^{i(\Phi_{14}-\Phi_{24})t}\right)\left( 1+e^{i(\Phi_{15}-\Phi_{25})t}\right)\left( 1+e^{i(\Phi_{16}-\Phi_{26})t}\right),\nonumber\\
    &\left\vert\lambda_{24}\right\vert^2\nonumber\\
    &=4\cos^2\left(\frac{\Phi_{14}-\Phi_{24}}{2}t\right)4\cos^2\left(\frac{\Phi_{15}-\Phi_{25}}{2}t\right)4\cos^2\left(\frac{\Phi_{16}-\Phi_{26}}{2}t\right).
\end{align}

For arbitrary $N$ and $k$, we can calculate $\left\vert\lambda_{xy}\right\vert^2$ for any $x<y\in[0,2^k-1]$ using the following algorithm. The correctness of this algorithm can be verified through detailed case-by-case calculations, as demonstrated in the preceding discussion.\\
For a given $\lambda_{xy}=\sum_{j_{k+1},\ldots,j_N}e^{i(\phi_{\mathcal{B}(x)}j_{k+1}\ldots j_N-\phi_{\mathcal{B}(y)}j_{k+1}\ldots j_N)t}$, consider the $a$-th bit $j_a$ of $\mathcal{B}(x)$ and $j'_a$ of $\mathcal{B}(y)$, where $1\leq a\leq k$. Define $S_a=j'_a-j_a$. Therefore, $S_a=+1$ if $j_a=0_a,~j'_a=1_a$, $S_a=-1$ if $j_a=1_a,~j'_a=0_a$, and $S_a=0$ if $j_a=j'_a=0_a,1_a$. So, for a given $\lambda_{xy}$, we get a set of $k$ `$S_a$'s. It is important to note that the condition $x<y$ makes the first nonzero element of the set to be $+1$.\footnote{This is because, the first nonzero element $S_a$ of the set corresponds to the first pair of $(j_a,j'_a)$, for which $j'_a\neq j_a$, and $\mathcal{B}(y)>\mathcal{B}(x)\Rightarrow j'_a>j_a$, $\therefore S_a=+1$.}
 $\left\vert\lambda_{xy}\right\vert^2$ is then given by,
\begin{align}
    \left\vert\lambda_{xy}\right\vert^2 = 2^{2(N-k)}\prod_{b=k+1}^N \cos^2\left(\frac{\sum_{a=1}^{k} S_a\Phi_{ab}}{2}t \right).
\end{align}

For example, in the $N=6,k=3$ case, the set of `$S_a$'s for $\lambda_{24}$ is $\{S_1=+1,S_2=-1,S_3=0\}$,
    \begin{table}[h!]
        \centering
        \begin{tabular}{|c||c c c|}
            \hline
             $\mathcal{B}(2)$&$0_1$&$1_2$&$0_3$ \\
             $\mathcal{B}(4)$&$1_1$&$0_2$&$0_3$ \\
             \hline
             &$S_1=+1$,&$S_2=-1$,&$S_3=0$\\
             \hline
        \end{tabular}
    \end{table}
    \\
and $\left\vert\lambda_{24}\right\vert^2$ is given by
\begin{align}
    &\left\vert\lambda_{24}\right\vert^2\nonumber\\
    &= 2^6\prod_{b=4}^6 \cos^2\left(\frac{\sum_{a=1}^{3} S_a\Phi_{ab}}{2}t \right) \nonumber\\
    &= 4\cos^2\left(\frac{\Phi_{14}-\Phi_{24}}{2}t \right) 4\cos^2\left(\frac{\Phi_{15}-\Phi_{25}}{2}t \right) 4\cos^2\left(\frac{\Phi_{16}-\Phi_{26}}{2}t \right). \nonumber
\end{align}

Now, note that different $\lambda_{xy}$s, may have the same set of `$S_a$'s, and hence same values of $\left\vert\lambda_{xy}\right\vert^2$. For example, $\lambda_{24}$ and $\lambda_{35}$ have the same set of `$S_a$'s,
\begin{table}[h!]
    \centering
    \begin{tabular}{|c||c c c|c|c||c c c|}
        \hline
        $\mathcal{B}(2)$&$0_1$&$1_2$&$0_3$ && $\mathcal{B}(3)$&$0_1$&$1_2$&$1_3$ \\
        $\mathcal{B}(4)$&$1_1$&$0_2$&$0_3$ && $\mathcal{B}(5)$&$1_1$&$0_1$&$1_3$ \\
        \hline
         &$S_1=+1$,&$S_2=-1$,&$S_3=0$ && &$S_1=+1$,&$S_2=-1$,&$S_3=0$\\
        \hline
    \end{tabular}
\end{table}
\\
and hence, $\left\vert\lambda_{24}(t)\right\vert^2=\left\vert\lambda_{35}(t)\right\vert^2$.
In particular, for a set of `$S_a$'s with $k'\leq k$ nonzero elements, there exists $2^{(k-k')}$ different $\lambda_{xy}$s that have the same set \footnote{Note here that, when we are saying two sets are same, order of the elements of the set is important, i.e., $\lbrace S_1=+1,S_2=-1,S_3=0 \rbrace \neq \lbrace S_1=+1,S_2=0,S_3=-1 \rbrace \neq \lbrace S_1=0,S_2=+1,S_3=-1 \rbrace$.}. So, while summing all the $\left\vert\lambda_{xy}\right\vert^2$s, we can collect the terms having same values as
\begin{align}\label{abc}
    \sum_{y=0}^{2^k-1}\sum_{x=0}^{y-1} &\left\vert\lambda_{xy}\right\vert^2\nonumber\\
    =\frac{2^{2N}}{2^{2k}}&\left\lbrace 2^{(k-1)}\sum_{a=1}^k \left[\prod_{b=k+1}^N \cos^2\left(\frac{\Phi_{ab}}{2}t \right)\right]+2^{(k-2)}\sum_{\substack{a_1<a_2\\a_1,a_2=1}}^k \left[\prod_{b=k+1}^N \cos^2\left(\frac{(\Phi_{a_1b}+\Phi_{a_2b})}{2}t \right)\right]\right.\nonumber\\
    &+\left.2^{(k-2)}\sum_{\substack{a_1<a_2\\a_1,a_2=1}}^k \left[\prod_{b=k+1}^N \cos^2\left(\frac{(\Phi_{a_1b}-\Phi_{a_2b})}{2}t \right)\right]\right.\nonumber\\
    &+\left.2^{(k-3)}\sum_{\substack{a_1<a_2<a_3\\a_1,a_2,a_3=1}}^k \left[\prod_{b=k+1}^N \cos^2\left(\frac{(\Phi_{a_1b}+\Phi_{a_2b}+\Phi_{a_3b})}{2}t \right)\right]\right.\nonumber\\
    &+\left.2^{(k-3)}\sum_{\substack{a_1<a_2<a_3\\a_1,a_2,a_3=1}}^k \left[\prod_{b=k+1}^N \cos^2\left(\frac{(\Phi_{a_1b}+\Phi_{a_2b}-\Phi_{a_3b})}{2}t \right)\right]\right.\nonumber\\
    &+\left.2^{(k-3)}\sum_{\substack{a_1<a_2<a_3\\a_1,a_2,a_3=1}}^k \left[\prod_{b=k+1}^N \cos^2\left(\frac{(\Phi_{a_1b}-\Phi_{a_2b}+\Phi_{a_3b})}{2}t \right)\right]+\ldots\right.\nonumber\\
    &+\left.2^{(k-k)}\sum_{\substack{a_1<a_2<\ldots<a_k\\a_1,a_2,\ldots,a_k=1}}^k \left[\prod_{b=k+1}^N \cos^2\left(\frac{(\Phi_{a_1b}-\Phi_{a_2b}-\ldots-\Phi_{a_kb})}{2}t \right)\right]\right\rbrace.
\end{align}
Here, the first set of terms come from those $\left\vert\lambda_{xy}\right\vert^2$s for which, only one `$S_a$' is nonzero and $+1$. There are ${k\choose1}=k$ possibilities where one among $k$ `$S_k$'s is $+1$. The summation over $a$ takes into account all those possibilities. The second and third sets of terms come from those $\left\vert\lambda_{xy}\right\vert^2$s for which, two `$S_a$'s are nonzer. For the second set, both `$S_a$'s are $+1$, while for the third set, the first nonzero `$S_a$' is $+1$ and the second nonzero `$S_a$' is $-1$. For each of these sets, there are $k\choose2$ such possibilities. ``$\sum_{\substack{a_1<a_2\\a_1,a_2=1}}^k$" counts those possibilities. Same goes for the other sets of terms.

We can write (\ref{abc}) in a sightly compact form as
\begin{align}
    &\sum_{y=0}^{2^k-1}\sum_{x=0}^{y-1} \left\vert\lambda_{xy}\right\vert^2\nonumber\\
    &=\frac{2^{2N}}{2^k}\left\lbrace \frac{1}{2} \sum_{a=1}^k \left[\prod_{b=k+1}^N \cos^2\left(\frac{\Phi_{ab}}{2}t \right)\right]\right.\nonumber\\
    &+\left.\frac{1}{2^2}\sum_{s_1=0}^{1}\sum_{\substack{a_1<a_2\\a_1,a_2=1}}^k \left[\prod_{b=k+1}^N \cos^2\left(\frac{(\Phi_{a_1b}+(-1)^{s_1}\Phi_{a_2b})}{2}t \right)\right]\right.\nonumber\\
    &+\left.\frac{1}{2^3}\sum_{s_1,s_2=0}^1\sum_{\substack{a_1<a_2<a_3\\a_1,a_2,a_3=1}}^k \left[\prod_{b=k+1}^N \cos^2\left(\frac{(\Phi_{a_1b}+(-1)^{s_1}\Phi_{a_2b}+(-1)^{s_2}\Phi_{a_3b})}{2}t\right)\right]+\ldots\right.\nonumber\\
    &+\left.\frac{1}{2^k}\sum_{s_1,\ldots ,s_{(k-1)}=0}^1\sum_{\substack{a_1<a_2<\ldots<a_k\\a_1,a_2,\ldots,a_k=1}}^k \left[\prod_{b=k+1}^N \cos^2\left(\frac{(\Phi_{a_1b}+(-1)^{s_1}\Phi_{a_2b}+\ldots+(-1)^{s_{(k-1)}}\Phi_{a_kb})}{2}t \right)\right]\right\rbrace.
\end{align}
Putting this expression in (\ref{pIcon}), we get the final expression for the I-concurrence as
\begin{align}
    &\mathcal{C}_{12\ldots k|(k+1)\ldots N}(t)\nonumber\\
    &=\left[\frac{2^k-1}{2^{k-1}}\right.\nonumber\\
    &-\left.\frac{4}{2^k}\left\lbrace \frac{1}{2} \sum_{a=1}^k \left[\prod_{b=k+1}^N \cos^2\left(\frac{\Phi_{ab}}{2}t \right)\right]\right.\right.\nonumber\\
    &+\left.\left.\frac{1}{2^2}\sum_{s_1=0}^{1}\sum_{\substack{a_1<a_2\\a_1,a_2=1}}^k \left[\prod_{b=k+1}^N \cos^2\left(\frac{(\Phi_{a_1b}+(-1)^{s_1}\Phi_{a_2b})}{2}t \right)\right]\right.\right.\nonumber\\
    &+\left.\left.\frac{1}{2^3}\sum_{s_1,s_2=0}^1\sum_{\substack{a_1<a_2<a_3\\a_1,a_2,a_3=1}}^k \left[\prod_{b=k+1}^N \cos^2\left(\frac{(\Phi_{a_1b}+(-1)^{s_1}\Phi_{a_2b}+(-1)^{s_2}\Phi_{a_3b})}{2}t\right)\right]+\ldots\right.\right.\nonumber\\
    &+\left.\left.\frac{1}{2^k}\sum_{s_1,\ldots ,s_{(k-1)}=0}^1\sum_{\substack{a_1<a_2<\ldots<a_k\\a_1,a_2,\ldots,a_k=1}}^k \left[\prod_{b=k+1}^N \cos^2\left(\frac{(\Phi_{a_1b}+(-1)^{s_1}\Phi_{a_2b}+\ldots+(-1)^{s_{(k-1)}}\Phi_{a_kb})}{2}t \right)\right]\right\rbrace\right]^{\frac{1}{2}}.
\end{align}

I-concurrences for the other possible $k|(N-k)$ bipartitions are given by re-labeling the masses (hence the corresponding entangling phases) in the above expression.

\bibliographystyle{JHEP}
\bibliography{biblio.bib}

\end{document}